\documentclass[a4paper,11pt,onecolumn,draftcls]{IEEEtran}
\usepackage[utf8]{inputenc}
\usepackage{amsthm,amsmath,amssymb}
\usepackage[ruled,vlined]{algorithm2e}
\newtheorem{definition}{Definition}
\newtheorem{lemma}[definition]{Lemma}
\newtheorem{proposition}[definition]{Proposition}
\newtheorem{theorem}[definition]{Theorem}

\newtheorem{remark}[definition]{Remark}
\usepackage{tikz}
\usepackage{pgfplots}
\usepgfplotslibrary{fillbetween}
\usepackage[style=ieee]{biblatex}
\usepackage{subcaption}
\usepackage{float}
\SetAlgorithmName{Protocol}

\allowdisplaybreaks

\newcommand\recht\operatorname
\DeclareMathOperator*{\argmin}{arg\,min}

\title{The Privacy-Utility Tradeoff of Robust Local Differential Privacy}
\author{Milan Lopuhaä-Zwakenberg and Jasper Goseling}
\date{\today}

\begin{document}

\maketitle

\begin{abstract}
We consider data release protocols for data $X=(S,U)$, where $S$ is sensitive; the released data $Y$ contains as much information about $X$ as possible, measured as $\recht{I}(X;Y)$, without leaking too much about $S$. We introduce the Robust Local Differential Privacy (RLDP) framework to measure privacy. This framework relies on the underlying distribution of the data, which needs to be estimated from available data. Robust privacy guarantees are ensuring privacy for all distributions in a given set $\mathcal{F}$, for which we study two cases: when $\mathcal{F}$ is the set of all distributions, and when $\mathcal{F}$ is a confidence set arising from a $\chi^2$ test on a publicly available dataset. In the former case we introduce a new release protocol which we prove to be optimal in the low privacy regime. In the latter case we present four algorithms that construct RLDP protocols from a given dataset. One of these approximates $\mathcal{F}$ by a polytope and uses results from robust optimisation to yield high utility release protocols. However, this algorithm relies on vertex enumeration and becomes computationally inaccessible for large input spaces. The other three algorithms are low-complexity and build on randomised response. Experiments verify that all four algorithms offer significantly improved utility over regular LDP.
\end{abstract}

% !TEX root =  main.tex

\section{Introduction}
We consider the setting in which users have data $X=(S,U)$ that a data aggregator is interested in, but users do not wish to disclose information about sensitive data $S$. Therefore, users release an obfuscated version $Y$ of $X$, such that $Y$ contains as much information about $X$ as possible, measured as $\recht{I}(X;Y)$, without leaking too much about $S$. This scenario and closely related ones have been studied in, for instance, ~\cite{rebollo2010t,kairouz2014extremal,makhdoumi2014information,salamatian2015managing,asoodeh2016information,kung2018compressive,ding2019submodularity,salamatian2020privacy}.
%Figure~\ref{fig:overview} gives an overview of the scenario. 

This paper introduces a form of local differential privacy (LDP)~\cite{duchi2013local} to measure the amount of information that $Y$ leaks on $S$. The following version of $\varepsilon$-LDP was introduced in~\cite{lopuhaa2020privacy}:
%\begin{equation*}
%\mathbb{P}(Y = y|X = x) \leq \textrm{e}^{\varepsilon}(Y = y|X = x').
%\end{equation*}
%Since we only need to protect $S$, this can be relaxed into
\begin{equation} \label{eq:intro_rldp}
\mathbb{P}(Y = y|S = s) \leq \textrm{e}^{\varepsilon}\mathbb{P}(Y = y|S = s'),
\end{equation}
for all $y$, $s$ and $s'$. Note, that this condition is less strict than $\mathbb{P}(Y = y|X = x) \leq \textrm{e}^{\varepsilon}\mathbb{P}(Y = y|X = x')$ as would be used in ordinary LDP. Also note that~\eqref{eq:intro_rldp} relies on the distribution $P_X=P_{S,U}$. From these observations it follows that this privacy definition enables higher utility of the released data $Y$ at the expense of not being completely `distribution free' as would be the case for ordinary LDP.

In~\cite{lopuhaa2020privacy} condition~\eqref{eq:intro_rldp} is studied for the case of known $P_{X}$. This is a strong assumption, since users will need to estimate $P_X$. When an attacker has better knowledge of $P_X$ than the user, it follows from the odds-ratio interpretation of differential privacy~\cite{kifer2014pufferfish} that sufficient privacy is not guaranteed in such a scenario.
%In~\cite{lopuhaa2020privacy} condition~\eqref{eq:intro_rldp} is studied for the case of known $P_{X}$. This is a strong assumption, since users will need to estimate $P_X$. More importantly, from the odds-ratio interpretation of differential privacy~\cite{kifer2014pufferfish} it follows that in this case privacy is guaranteed only with respect to attackers that know $P_X$. In practice, the side-information available to an attacker will often be partial information about $P_X$.   \textbf{(I do not understand this. Can an attacker with less information have a better attack?)}

In this paper we, therefore, provide stronger privacy guarantees. In particular, we introduce robustness constraints, which say that privacy should not just hold for one $P_{X}$, but for a set $\mathcal{F}$ of these.
% In particular, we assume that there is publicly available data from $n$ users, which allows the users as well as an attacker to estimate $\hat{P}_{X}$. The set $\mathcal{F}$ consists of those $P$ that are close enough to $\hat{P}_X$ so that the difference is not statistically significant for a chosen significance level $\alpha$.
%; this choice of $\mathcal{F}$ is common in statistical optimisation. 
%At the same time, we introduce a variant of local differential privacy (LDP) to measure the information leakage of $S$ in $Y$.
As a result we guarantee privacy against attackers with (at least) reasonable estimates of $P_X$, without sacrificing utility to protect against attackers with no or unreliable information on $P_X$. We refer to the resulting privacy framework as Robust Local Differential Privacy (RLDP).

We consider two cases for $\mathcal{F}$. In the first case, we let $\mathcal{F}$ be the set of \emph{all} probability distributions $\mathcal{F}$. We show that in this case, privacy w.r.t. $S$ is very similar, but not equivalent, to privacy w.r.t. $X$. We introduce a new privacy protocol that exploits the small difference that remains between these two definitions and show that this protocol is optimal in the low privacy regime.

In the second case, we assume that there is publicly available data from $n$ users, which allows the aggregator and the users to estimate $\hat{P}_{X}$. The set $\mathcal{F}$ consists of those $P$ that are close enough to $\hat{P}$ so that the difference is not statistically significant for a chosen significance level $\alpha$; this choice of $\mathcal{F}$ is common in statistical optimisation. Here, we introduce three protocols and study their privacy and utility.

\subsection{Contributions}
In addition to introducing the RLDP privacy framework, the main contributions of this paper are as follows.
%the introduction of protocols that offer robust local differential privacy in the Privacy Funnel scenario. These come in three forms:
%In the setting where $\mathcal{F}$ is a confidence interval around $\hat{P}_X$, w

\noindent We consider the setting where $\mathcal{F} = \mathcal{P}_{\mathcal{X}}$. In this setting:
\begin{itemize}
\item We introduce a protocol SRR based on the classic Randomized Response protocol \cite{warner1965randomized}. We show that SRR maximises mutual information in the low privacy regime.
\end{itemize}

\noindent We consider the setting where $\mathcal{F}$ is a $\chi^2$ confidence set around $\hat{P}_X$. In this setting:
\begin{itemize}
\item We approximate $\mathcal{F}$ by an enveloping polytope.
%We characterize $\mathcal{F}$ by an enveloping polytope.
We then use techniques from robust optimisation~\cite{ben2009robust,ben2015deriving,bertsimas2018data} to characterize the protocol that is optimal over this polytope. The resulting lower bound on utility demonstrates the advantage of RLDP over ordinary LDP. A drawback of this method is that it relies on vertex enumeration and is, therefore, computationally unfeasible for large alphabets.
\item Therefore, we introduce two low-complexity data release mechanisms: i) Independent Reporting (IR), in which $S$ and $U$ are reported through separate LDP protocols, and ii) Conditional Reporting (CR), in which first $S$ is obfuscated, and either a slightly obfuscated $U$ or a randomly drawn $U'$ is returned, depending on whether the obfuscated $S$ is `correct'.
\item For both mechanisms we characterize the conditions that underlying LDP protocols have to satisfy in order to ensure RLDP. Furthermore, while both mechanisms can incorporate any LDP protocol, we show that it is optimal to use Randomised Response~\cite{warner1965randomized}. This drastically reduces the search space and allows us to find the optimal SR and CR mechanisms using one-dimensional optimisation.
\end{itemize}

We demonstrate the improved utility of RLDP over LDP with numerical experiments. In particular we provide results for both synthetic datasets as well as real-world census data.

\subsection{Related work}
Disclosing information in a privacy-preserving way is one of the main challenges in official statistics~\cite{willenborg2012elements, hundepool2012statistical}. The setting considered in the current paper close connected to disclosing a table with micro-data where each record in the table is released independently of the other records. This approach to disclosing micro-data was studied in~\cite{rebollo2010t} by considering expected error as the utility measure and mutual information as the privacy measure. The resulting optimization problem corresponds to the traditional rate-distortion problem. 

The version of the problem in which both utility and privacy are measured using mutual information is known as the privacy funnel and was studied first in~\cite{makhdoumi2014information}. The dual problem of the privacy funnel, in which utility is maximized w.r.t.\ a privacy constraint was studied in~\cite{asoodeh2016information}. The privacy funnel and its dual are intimately related to the information bottleneck problem~\cite{tishby2000information}, which seeks to optimise compression while retaining relevant information. Multiple approaches to optimising privacy funnel also work for the information bottleneck and vice versa \cite{ding2019submodularity,kung2018compressive}.

In~\cite{salamatian2015managing} a version of this problem is studied in which privacy leakage is measured through the improved statistical inference by an attacker after seeing the disclosed information. This measure is formulated through a general cost function, with mutual information resulting as a special case. Perfect privacy, which demands the output to be independent of the sensitive data, is studied in \cite{rassouli2018perfect}, and methods are given to find optimal protocols in this setting. In~\cite{issa2019operational} the maximal leakage measure with a clear operational interpretation is defined. In~\cite{liao2019tunable} this measure is generalized to a parametrized measure, enabling to interpolate between maximal leakage and mutual information. A multitude of other privacy frameworks and leakage measures exist. We refer to~\cite{wagner2018technical} for an overview and restrict the remainder of this section to local differential privacy and robustness, which are most closely related to our work. 

In this paper we consider measures based on Local Differential Privacy (LDP)~\cite{kasiviswanathan2011can,duchi2013local}. In this setting, several privacy protocols exist, including Randomised Response \cite{warner1965randomized} and Unary Encoding \cite{wang2017locally}. Optimal LDP protocols under a variety of utility metrics, including mutual information, are found in \cite{kairouz2014extremal}. A variation of LDP is proposed in~\cite{lopuhaa2020privacy} for the case of disclosing $X=(S,U)$, where only $S$ is sensitive. The privacy metrics given there fit into a general framework called pufferfish privacy~\cite{kifer2014pufferfish}. In \cite{salamatian2020privacy} a general class of privacy metrics called \emph{average information leakage} is introduced in this setting, and it is shown that LDP implies privacy under these metrics.

In all the above work the privacy protocol is derived from the (estimated) distribution $P_{S,X}$. In most cases an analysis of robustness/sensitivity with respect to this estimate is not present. An exception is~\cite{salamatian2015managing} in which one of the contributions is to quantify the impact of mismatched priors, i.e. the impact of not knowing $P_{S,X}$ exactly. A bound on the resulting level of privacy is derived in terms of the total variational distance between the actual and the estimated $P_{S,X}$. The behaviour of privacy and utility metrics under robustness are studied in \cite{wang2018utility,diaz2019robustness}. For a wide variety of privacy and utility metrics, they give bounds on the utility loss that occurs when robustness is added to the requirements. In both cases, robustness is defined by looking at an $\ell_1$-ball around the observed empirical distribution. One can also define robustness in other ways, such as by KL-divergence \cite{wang2016likelihood}, $\chi^2$-divergence \cite{bertsimas2018data}, or a general $f$-divergence \cite{duchi2016statistics}.

Another line of work builds on recent advances in generative adversarial networks~\cite{goodfellow2020generative}. In~\cite{huang2017context, tripathy2019privacy} the generative adversarial framework is used to provide release protocols that do not use explicit expressions for $P_X$. Even though it is not explitly addressed in~\cite{huang2017context, tripathy2019privacy}, it is expected that the generalization properties of networks will provide a form of robustness. Closely related approaches are used the area of face recognition, \cite{mirjalili2018semi, bortolato2020learning} with the aim of preventing biometric profiling~\cite{stoker2016facial}. In~\cite{mirjalili2018semi, bortolato2020learning}, however, the leakage measures that are used do not seem to have an operational interpretation.

\subsection{Overview of paper}
The structure of this paper is as follows. In Section \ref{sec:model} we describe the model in detail. In Section \ref{sec:maxf} we consider the case that $\mathcal{F} = \mathcal{P}_{\mathcal{X}}$. In Section \ref{sec:prop} we study the case that $\mathcal{F}$ is a confidence set, and we prove several properties of $\mathcal{F}$ that will be useful in the following sections. In Section \ref{sec:poly} we introduce PolyOpt, an algorithm that finds high utility protocols through approximating $\mathcal{F}$ by an enveloping polytope. In Section \ref{sec:sep} we discuss Independent Reporting, its privacy and utility, and we show how the optimal IR-protocol can be found using low-dimensional optimisation. In Section \ref{sec:cond} we do the same for Conditional Reporting. In Section \ref{sec:exp} we evaluate the discussed methods experimentally. Finally, in Section~\ref{sec:disc} we provide a discussion of our results and provide an outlook on future work.

\section{Model and Preliminaries} \label{sec:model}

The dataspace is $\mathcal{X} = \mathcal{S} \times \mathcal{U}$, where $\mathcal{S}$ and $\mathcal{U}$ are finite sets. We write $|\mathcal{S}| =: a_1$, $|\mathcal{U}| =: a_2$, and $|\mathcal{X}| = a_1a_2 =: a$.
New data items $X = (S,U)$ are drawn from a probability distribution $P^*$ in $\mathcal{P}_{\mathcal{X}}$, the space of probability distributions on $\mathcal{X}$. The user's aim is to create a release protocol $\mathcal{Q}$ such that $Y = \mathcal{Q}(X)$ contains as much information about $X$ as possible, while not leaking too much information about $S$. Protocol $\mathcal{Q}$ is a probabilistic map, that we represent by a left stochastic matrix $(Q_{y|x})_{y \in \mathcal{Y},x \in \mathcal{X}}$, and we write $|\mathcal{Y}| = b$. Often, we identify $\mathcal{Y} = \{1,\ldots,b\}$, and likewise for other sets.

The distribution $P^*$ is not known exactly. Instead it is known only that $P^*\in\mathcal{F}$ for some set of possible distributions $\mathcal{F}\subset\mathcal{P}_{\mathcal{X}}$, where $\mathcal{P}_{\mathcal{X}}$ denotes the probability simplex over $\mathcal{X}$. We give various examples of such $\mathcal{F}$ below. The uncertainty set $\mathcal{F}$ captures our uncertainty about $P^*$. The idea is that we guarantee privacy for all $P\in\mathcal{F}$. We will denote this as robust local differential privacy (RLDP).

\begin{definition}
Let $\varepsilon \geq 0$ and $\mathcal{F}\subset\mathcal{P}_{\mathcal{X}}$. We say that $\mathcal{Q}$ satisfies $(\varepsilon, \mathcal{F})$-RLDP if for all $s,s' \in \mathcal{S}$, all $y \in \mathcal{Y}$, and all $P \in \mathcal{F}$ we have
\begin{equation} \label{eq:rpp}
\mathbb{P}_{X \sim P}(Y = y|S = s) \leq \textrm{\emph{e}}^{\varepsilon}\mathbb{P}_{X \sim P}(Y = y|S = s').
\end{equation}
\end{definition}
Note that we use the notation $\mathbb{P}_{X\sim P}(\bullet)$ to emphasize that $X$ is distributed according to $P$. If no confusion can arise, we will often leave out the subscript $X\sim P$ to improve readability. 

We consider various forms of uncertainty on $P^*$, as captured by $\mathcal{F}$: 
\begin{enumerate}
    \item Nothing is known about $P^*$. In this case $\mathcal{F} = \mathcal{P}_{\mathcal{X}}$. Regarding privacy, this is the `safest' choice.
    \item We suppose there is a database $\vec{x} = (x_1,\cdots,x_n)$ accessible to the user, where each $x_i = (s_i,u_i)$ is drawn independently from $P^*$. Based on this, the user produces an estimate $\hat{P}$ of $P$. Fix a significance level $\alpha$: we let $\mathcal{F}$ be the $(1-\alpha)$-confidence interval for $P$ in a $\chi^2$-test, i.e.
    \begin{equation} \label{eq:chi2}
    \mathcal{F} = \left\{P: \sum_x \frac{(\hat{P}_x-P_x)^2}{P_x} \leq B := \frac{F^{-1}_{\#\mathcal{X}-1}(1-\alpha)}{n}\right\},
    \end{equation}
    where $F_d$ is the cdf of the $\chi^2$-distribution with $d$ degrees of freedom. At times, it will be convenient to express this as 
    \begin{equation} \label{eq:chi22}
    \mathcal{F} = \left\{P: \sum_x \frac{\hat{P}_x^2}{P_x} \leq B+1\right\}.
    \end{equation}
    This situation is expressed in Figure \ref{fig:overview}.
\end{enumerate}

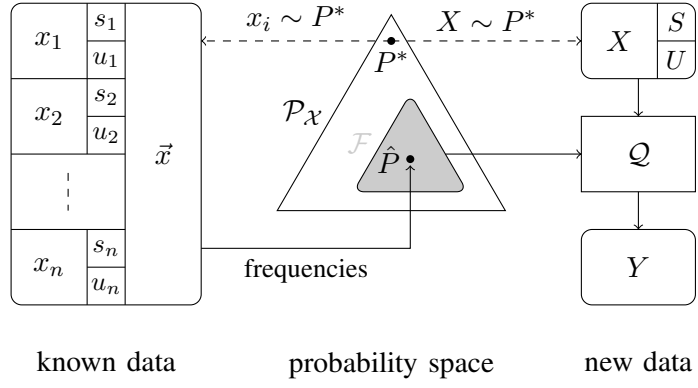
\begin{figure}
    \centering
\begin{tikzpicture}[scale=0.5]

\draw[rounded corners] (-7,0.5) -- (-12,0.5) -- (-12,-7.5) -- (-7,-7.5) -- cycle;
\draw[-] (-9,0.5) -- (-9,-7.5);
\draw[-] (-9,-1.5) -- (-12,-1.5);
\draw[-] (-9,-3.5) -- (-12,-3.5);
\draw[-] (-9,-5.5) -- (-12,-5.5);
\draw (-8,-3.5) node {$\vec{x}$};
\draw (-11,-0.5) node(a1) {$x_1$};
\draw (-9.5,0) node{\small $s_1$};
\draw (-9.5,-1) node{\small $u_1$};
\draw[-] (-10,-1.5) -- (-10,0.5);
\draw[-] (-9,-0.5) -- (-10,-0.5);
\draw (-11,-2.5) node(a1) {$x_2$};
\draw (-9.5,-2) node{\small $s_2$};
\draw (-9.5,-3) node{\small $u_2$};
\draw[-] (-10,-3.5) -- (-10,-1.5);
\draw[-] (-9,-2.5) -- (-10,-2.5);
\draw (-11,-6.5) node(a1) {$x_n$};
\draw (-9.5,-6) node{\small $s_n$};
\draw (-9.5,-7) node{\small $u_n$};
\draw[-] (-10,-5.5) -- (-10,-7.5);
\draw[-] (-9,-6.5) -- (-10,-6.5);

\draw[-, dashed] (-10.5,-4) -- (-10.5,-5);

%\draw[-latex,decorate,decoration={snake}] (3,-1) --node[above]{$X_i \sim P^*$} (7,-1);
%\draw[-latex] (7,-5) --node[above]{frequencies} (3,-5);
\draw (-9.5,-9) node {known data};
\draw (-2,-9.1) node{probability space};
\draw (4.5,-9) node {new data};
\draw (1,-5) -- (-5,-5) --node[left]{$\mathcal{P}_{\mathcal{X}}$} (-2,0.1961524227) -- cycle;
\draw[rounded corners,text = black, draw = black, fill = black, fill opacity = 0.2, text opacity = 1] (0,-4.5) -- (-3,-4.5) --node[left]{$\mathcal{F}$} (-1.5,-1.90192378865) -- cycle;
\fill (-1.5,-3.63397459622) circle[radius=3pt] node[left]{$\hat{P}$};
\fill (-2,-0.5) circle[radius=3pt] node[below]{$P^*$};
\draw[->] (-7,-6) --node[below]{\small frequencies} (-1.5,-6) -- (-1.5,-3.8);
\draw[->,dashed] (-2,-0.5) --node[above]{$x_i \sim P^*$} (-7,-0.5);
\draw[->,dashed] (-2,-0.5) --node[above]{$X \sim P^*$} (3,-0.5);
\draw[rounded corners] (3,0.5) -- (6,0.5) -- (6,-1.5) -- (3,-1.5) -- cycle;
\draw (4,-0.5) node{$X$};
\draw (5,-1.5) -- (5,0.5);
\draw (5,-0.5) -- (6,-0.5);
\draw (5.5,0) node{\small $S$};
\draw (5.5,-1) node{\small $U$};
\draw (6,-2.5) -- (3,-2.5) -- (3,-4.5) -- (6,-4.5) -- cycle;
\draw (4.5,-3.5) node{$\mathcal{Q}$};
\draw[->] (-0.57735026919,-3.5) -- (3,-3.5);
\draw[rounded corners] (3,-5.5) -- (6,-5.5) -- (6,-7.5) -- (3,-7.5) -- cycle;
\draw (4.5,-6.5) node{$Y$};
\draw[->] (4.5,-1.5) -- (4.5,-2.5);
\draw[->] (4.5,-4.5) -- (4.5,-5.5);
\end{tikzpicture}
    \caption{An overview of the setting of this paper when $\mathcal{F}$ is a confidence set based on a dataset $\vec{x}$.}
    \label{fig:overview}
\end{figure}

Another option would be to have $\mathcal{F}$ be a singleton, i.e. to assume that $P$ is known. This setting is studied in \cite{lopuhaa2020privacy}.

For completeness we give the definition of LDP.
\begin{definition} \label{def:ldp}
Let $\varepsilon \geq 0$. We say that $\mathcal{Q}\colon \mathcal{X} \rightarrow \mathcal{Y}$ satisfies $\varepsilon$-LDP if for all $x,x' \in \mathcal{X}$ and all $y \in \mathcal{Y}$ we have
\begin{equation}
\mathbb{P}(Y = y|X = x) \leq \textrm{\emph{e}}^{\varepsilon}\mathbb{P}(Y = y|X = x').
\end{equation}
\end{definition}

In Sections \ref{sec:sep} and \ref{sec:cond}, we build RLDP protocols from regular LDP protocols. To establish the privacy guarantees of these protocols, we will need the following lemma that relates LDP to the $\ell_1$-distance of probability distributions.

\begin{lemma} \label{lem:ldp}
Let $\mathcal{Q}\colon \mathcal{X} \rightarrow \mathcal{Y}$ be an $\varepsilon$-LDP protocol. Then for all $y \in \mathcal{Y}$ and all $P,P' \in \mathcal{P}_{\mathcal{X}}$ we have
\begin{equation}
\frac{\mathbb{P}_{X \sim P}(\mathcal{Q}(X) = y)}{\mathbb{P}_{X \sim P'}(\mathcal{Q}(X) = y)} \leq 1+\frac{\textrm{\emph{e}}^{\varepsilon}-1}{2}||P-P'||_1.
\end{equation}
\end{lemma}

\begin{proof}
Let $Q^{\recht{max}}_y = \max_{x} Q_{y|x}$ and $Q^{\recht{min}}_y = \min_{x} Q_{y|x}$; note that $Q^{\recht{max}}_y \leq \textrm{e}^{\varepsilon}Q^{\recht{min}}_y$. Furthermore,  $\mathbb{P}_{X \sim P}(\mathcal{Q}(X) = y) = \sum_{x \in \mathcal{X}} Q_{y|x}P_x$ and $\mathbb{P}_{X \sim P'}(\mathcal{Q}(X) = y) = \sum_{x \in \mathcal{X}} Q_{y|x}P'_x$, hence
\begin{align}
& \mathbb{P}_{X \sim P}(\mathcal{Q}(X) = y)-\mathbb{P}_{X \sim P'}(\mathcal{Q}(X) = y)\\
&= \sum_{x: P_x \geq P'_x} Q_{y|x}(P_x-P'_x) - \sum_{x: P'_x > P_x} Q_{y|x}(P'_x-P_x)\\
&\leq \frac{Q^{\recht{max}}_y}{2}||P-P'||_1 - \frac{Q^{\recht{min}}_y}{2}||P-P'||_1 \\
&\leq \frac{(\textrm{e}^{\varepsilon}-1)Q^{\recht{min}}_y}{2}||P-P'||_1 \\
&\leq \frac{(\textrm{e}^{\varepsilon}-1)\mathbb{P}_{X \sim P'}(\mathcal{Q}(X) = y)}{2}||P-P'||_1,
\end{align}
from which the lemma directly follows.
\end{proof}

%\subsection{Utility}
Next to a privacy leakage measure we need to define a utility measure. Throughout this paper, we follow the original Privacy Funnel \cite{makhdoumi2014information} and its LDP counterpart \cite{lopuhaa2020privacy} in taking mutual information $\recht{I}(X;Y)$ as a utility measure. As is argued in \cite{makhdoumi2014information}, mutual information arises naturally when minimising log loss distortion in the Privacy Funnel scenario.

The value of $\recht{I}(X;Y)$ depends on $\mathcal{Q}$ and on the probability distribution on $\mathcal{X}$. As this is unknown, we consider two possibilities:
\begin{enumerate}
    \item One can take $\recht{I}_{X \sim \hat{P}}(X;Y)$, abbreviated to $\recht{I}_{\hat{P}}(X;Y)$;
    \item One can consider $\min_{P \in \mathcal{F}} \recht{I}_P(X;Y)$.
\end{enumerate}

Throughout this paper, all results will be proven for general $P$. Furthermore, it will turn out that many protocols we find will not depend on $P$. In the experiments of Section \ref{sec:exp}, we focus on $\recht{I}_{\hat{P}}(X;Y)$, although we also investigate the effect of $P$ on utility by comparing $\recht{I}_{P^*}(X;Y)$ to $\recht{I}_{\hat{P}}(X;Y)$ in Section \ref{ssec:robust}.
\section{Maximal $\mathcal{F}$} \label{sec:maxf}

In this section, we consider the case where $\mathcal{F}$ is maximal, i.e. $\mathcal{F} = \mathcal{P}_{\mathcal{X}}$. We show that in this situation, RLDP is almost equivalent to LDP. However, it is not completely equivalent, and we use this to describe a version of Generalised Randomised Response (GRR) that exploits the difference between RLDP and LDP. We show that this new protocol is optimal in the low privacy regime (i.e. $\varepsilon \gg 0$), similar to how GRR is the optimal LDP-protocol in the low privacy regime \cite{kairouz2014extremal}. The following Proposition gives a characterisation of $(\varepsilon,\mathcal{P}_{\mathcal{X}})$-RLDP.

\begin{proposition}
$\mathcal{Q}$ satisfies $(\varepsilon, \mathcal{P}_{\mathcal{X}})$-RLDP if and only if for all $y \in \mathcal{Y}$ and $(s,u),(s',u') \in \mathcal{X}$ with $s \neq s'$ one has
\begin{equation}
\frac{Q_{y|s,u}}{Q_{y|s',u'}} \leq \textrm{\emph{e}}^{\varepsilon}.
\end{equation}
\end{proposition}

\begin{proof}
Suppose that $\mathcal{Q}$ satisfies $(\varepsilon, \mathcal{F})$-RLDP w.r.t. $\mathcal{P}_{\mathcal{X}}$. Let $(s,u),(s',u') \in \mathcal{X}$ with $s \neq s'$. Let $P$ be given by
\begin{equation}
    P_x = \left\{\begin{array}{ll}
    \tfrac{1}{2}, \ & \textrm{ if $x \in \{(s,u),(s',u')\}$},\\
    0, \ & \textrm{ otherwise}.
    \end{array}\right.
\end{equation}
Then
\begin{equation}
\frac{Q_{y|s,u}}{Q_{y|s',u'}} = \frac{\mathbb{P}(\mathcal{Q}(X) = y|S =s)}{\mathbb{P}(\mathcal{Q}(X) = y|S = s')} \leq \textrm{e}^{\varepsilon}.
\end{equation}
On the other hand, suppose that $\frac{Q_{y|s,u}}{Q_{y|s',u'}} \leq \textrm{e}^{\varepsilon}$ for all $s\neq s'$ and $u,u'$. Then for all $s \neq s'$ and $P$ we have
\begin{equation}
\frac{\mathbb{P}(\mathcal{Q}(X) = y|S =s)}{\mathbb{P}(\mathcal{Q}(X) = y|S = s')} = \frac{\sum_u Q_{y|s,u}P_{u|s}}{\sum_{u'} Q_{y|s',u'}P_{u'|s'}} \leq \textrm{e}^{\varepsilon}.
\end{equation}
Hence, $\mathcal{Q}$ satisfies $(\varepsilon,\mathcal{P}_{\mathcal{X}})$-RLDP w.r.t. $\mathcal{F}$.
\end{proof}

The proposition demonstrates that RLDP is very similar to LDP. The difference is that the condition ``for all $x, x'\in\mathcal{X}$'' from Definition~\ref{def:ldp} is relaxed to only those $x$ and $x'$ for which $s\neq s'$. We will exploit this difference. Recall that Generalised Randomised Response \cite{warner1965randomized} is the privacy protocol $\recht{GRR}^{\varepsilon}\colon\mathcal{X} \rightarrow \mathcal{X}$ given by

\begin{equation}
\recht{GRR}^{\varepsilon}_{y|x} = \left\{\begin{array}{ll}
\frac{\textrm{e}^{\varepsilon}}{\textrm{e}^{\varepsilon}+a-1} & \textrm{ if $x = y$,}\\
\frac{1}{\textrm{e}^{\varepsilon}+a-1} & \textrm{ otherwise}.
\end{array}\right.
\end{equation}

This protocol has been designed such that $\frac{\recht{GRR}^{\varepsilon}_{y|x}}{\recht{GRR}^{\varepsilon}_{y|x'}} = \textrm{e}^{\pm\varepsilon}$ for $x \neq x'$, the maximal fractional difference that $\varepsilon$-LDP allows. We will see that for RLDP we can go up to a difference of $\textrm{e}^{\pm 2\varepsilon}$ if $x = (s,u)$ and $x'= (s,u')$, as we typically only need to satisfy
\begin{equation} \label{eq:SRRintuition}
Q_{y|s,u} \leq \textrm{e}^{\varepsilon}Q_{y|s',u'} \leq \textrm{e}^{2\varepsilon}Q_{y|s,u'}.
\end{equation}
We capture the intuition from necessary condition~\eqref{eq:SRRintuition} in a new protocol called \emph{Secret Randomised Response (SRR)}.

\begin{definition}[Secret Randomised Response (SRR)]
Let $\varepsilon>0$. Then the release protocol $\recht{SRR}^{\varepsilon}\colon \mathcal{X} \rightarrow \mathcal{X}$ is given by
\begin{equation}
\recht{SRR}^{\varepsilon}_{s',u'|s,u} = \left\{\begin{array}{ll}
\frac{\textrm{\emph{e}}^{\varepsilon}}{\textrm{\emph{e}}^{\varepsilon}+\textrm{\emph{e}}^{-\varepsilon}(a_2-1)+a-a_2}, & \textrm{ if $(s',u') = (s,u)$,}\\
\frac{\textrm{\emph{e}}^{-\varepsilon}}{\textrm{\emph{e}}^{\varepsilon}+\textrm{\emph{e}}^{-\varepsilon}(a_2-1)+a-a_2}, & \textrm{ if $s' = s$ and $u' \neq u$,}\\
\frac{1}{\textrm{\emph{e}}^{\varepsilon}+\textrm{\emph{e}}^{-\varepsilon}(a_2-1)+a-a_2}, & \textrm{ if $s'\neq s$,}\\
\end{array}\right.
\end{equation}
\end{definition} 

The next result demonstrates that the necessary condition~\eqref{eq:SRRintuition} is, in the case of SRR, also sufficient.
\begin{lemma} \label{lem:sgrr}
SRR satisifes $(\varepsilon,\mathcal{\mathcal{P}_{\mathcal{X}}})$-RLDP.
\end{lemma}

\begin{proof}
It can be directly verified that for all $s \neq s'$, $u$, $u'$ and $y$ we have
$\frac{\recht{SRR}^{\varepsilon}_{y|s,u}}{\recht{SRR}^{\varepsilon}_{y|s',u'}} \in \{\textrm{e}^{-\varepsilon},1,\textrm{e}^{\varepsilon}\}$, from which $(\varepsilon,\mathcal{P}_{\mathcal{X}})$-RPP follows.
\end{proof}

As for utility, note that the robust utility metric $\min_{P \in \mathcal{F}} \recht{I}_{P}(X;Y)$ is not useful if $\mathcal{F} = \mathcal{P}_{\mathcal{X}}$, since by considering a degenerate $P$ it follows immediately that $\recht{I}_{P}(X;Y)=0$ for every $\mathcal{Q}$. However, SRR is optimal in the following sense:

\begin{theorem} \label{thm:sgrr}
For every $P$, there is a $\varepsilon_0 \geq 0$ such that for all $\varepsilon \gg \varepsilon_0$ such that SRR is the $(\varepsilon,\mathcal{\mathcal{P}_{\mathcal{X}}})$-RLDP protocol maximising $\recht{I}_P(X;Y)$.
\end{theorem}

The proof of this theorem follows along the same lines as the proof of Theorem 14 of \cite{kairouz2014extremal}, in which it is proven that GRR is the optimal LDP protocol for $\varepsilon$ large enough. The proof is presented in Appendix \ref{ssec:pfsgrr}.

\section{Properties of the domain $\mathcal{F}$} \label{sec:prop}

From this point onwards we consider $\mathcal{F}$ to be of the form in  (\ref{eq:chi2}). Before we introduce new algorithms in Sections \ref{sec:poly}--\ref{sec:cond}, we need some technical results on properties of $\mathcal{F}$. First some notation: for $u \in \mathcal{U}$ and $s \in \mathcal{S}$, we write
\begin{align}
P_{u} &= \sum_s P_{u,s},\\
P_{s} &= \sum_u P_{u,s},\\
P_{u|s} &= \frac{P_{u,s}}{P_s},\\
P_{\mathcal{U}|s} &= (P_{u|s})_{u \in \mathcal{U}} \in \mathcal{P}_{\mathcal{U}}.
\end{align}

The following lemma states that for every $s$, the image of $\mathcal{F}$ under the projection $P \mapsto P_{\mathcal{U}|s}$ is again of the form in (\ref{eq:chi2}).

\begin{lemma} \label{lem:bs}
Let $s \in \mathcal{S}$ such that $\hat{P}_s > 0$. Let $\mathcal{F}_{\mathcal{U}|s}$ be the projection of $\mathcal{F}$ onto $\mathcal{P}_{\mathcal{U}}$ via the map $P \mapsto P_{\mathcal{U}|s} \in \mathcal{P}_{\mathcal{U}}$. Define $B_{s} := \frac{(\sqrt{B+1}+\hat{P}_s-1)^2}{\hat{P}_s^2}-1$. Then
\begin{equation} \label{eq:dalphas}
\mathcal{F}_{\mathcal{U}|s} = \left\{R \in \mathcal{P}_{\mathcal{U}}: \sum_u \frac{(\hat{P}_{u|s}-R_u)^2}{R_u} \leq B_s\right\}.
\end{equation}
\end{lemma}

\begin{proof}
For $P \in \mathcal{F}$ and $s \in \mathcal{S}$ one has, using the definition of $\mathcal{F}$ in (\ref{eq:chi22}),
\begin{align}
\frac{\hat{P}^2_s}{P_s}\sum_u \frac{\hat{P}^2_{u|s}}{P_{u|s}} &= \sum_u \frac{\hat{P}^2_{s,u}}{P_{s,u}} \\
&\leq B+1-\sum_{s' \neq s}\sum_u \frac{\hat{P}_{s',u}^2}{P_{s',u}}\\
&= B+1-\frac{(1-\hat{P}_s)^2}{1-P_s}\sum_{s' \neq s}\sum_u\frac{\hat{P}_{s',u|\neg s}^2}{P_{s',u| \neg s}},
\end{align}
where for $s' \in \mathcal{S} \setminus \{s\}$ and $u \in \mathcal{U}$ we define $P_{s',u|\neg s} = \frac{P_{u,s'}}{1-P_s}$. These form a probability distribution on $(\mathcal{S}\setminus\{s\})\times \mathcal{U}$. As such we have
\begin{equation}
    \sum_{s' \neq s}\sum_u\frac{\hat{P}_{u,s'|\neg s}^2}{P_{u,s'| \neg s}} = 1+\sum_{s' \neq s}\sum_u\frac{(P_{u,s'|\neg s}-\hat{P}_{u,s'|\neg s})^2}{P_{u,s'| \neg s}} \geq 1.
\end{equation}
It follows that
\begin{equation}
\sum_u \frac{\hat{P}^2_{u|s}}{P_{u|s}} \leq \frac{P_s}{\hat{P}^2_s}\left(B+1-\frac{(1-\hat{P}_s)^2}{1-P_s}\right).
\end{equation}
We find the maximum of the right hand side by differentiating with respect to $P_s$, for which we get
\begin{equation}
\frac{B+1}{\hat{P}^2_s} -\frac{(1-\hat{P}_s)^2}{\hat{P}^2_s(1-P_s)^2}.
\end{equation}
Setting this equal to $0$ and solving w.r.t. $P_s$, we find that the maximum is attained at $P_s = 1-\frac{1-\hat{P}_s}{\sqrt{B+1}}$. Substituting this, we find
\begin{equation}
    \frac{P_s}{\hat{P}^2_s}\left(B+1-\frac{(1-\hat{P}_s)^2}{1-P_s}\right) \leq \frac{(\sqrt{B+1}-1+\hat{P}_s)^2}{\hat{P}_s^2} = B_s+1,
\end{equation}
hence $\sum_u \frac{(P_{u|s}-\hat{P}_{u|s})^2}{P_{u|s}} \leq B_s$; this shows the inclusion ``$\subset$'' in (\ref{eq:dalphas}). On the other hand, suppose that $R \in \mathcal{P}_{\mathcal{U}}$ satisfies $\sum_u \frac{\hat{P}_{u|s}^2}{R_u} \leq B_s+1$. Let $c = 1-\frac{1-\hat{P}_s}{\sqrt{B+1}}$, and define $P \in \mathcal{P}_{\mathcal{X}}$ by
\begin{equation}
P_{u,s'} = \left\{\begin{array}{ll}
cR_u, & \textrm{ if $s' = s$,} \\
\frac{\hat{P}_{u,s'}}{\sqrt{B+1}} & \textrm{ otherwise.}
\end{array}\right.
\end{equation}
Then $P_{\mathcal{U}|s} = R$, and
\begin{align}
\sum_{u,s'} \frac{\hat{P}_{u,s'}^2}{P_{u,s'}} &= \sum_u \frac{\hat{P}_{u,s}^2}{cR_u} + \sum_u \sum_{s' \neq s} \sqrt{B+1}\hat{P}_{u,s'}\\
&= \frac{\hat{P}^2_s\sqrt{B+1}}{\sqrt{B+1}-1+\hat{P}_s} \sum_u \frac{\hat{P}_{u|s}^2}{R_u} + \sqrt{B+1}(1-\hat{P}_s)\\
&\leq \frac{\hat{P}^2_s\sqrt{B+1}}{\sqrt{B+1}-1+\hat{P}_s}\cdot \frac{(\sqrt{B+1}-1+\hat{P}_s)^2}{\hat{P}_s^2} + \sqrt{B+1}(1-\hat{P}_s)\\
&=B+1,
\end{align}
hence $P \in \mathcal{F}$. This shows the inclusion ``$\supset$'' in (\ref{eq:dalphas}).
\end{proof}

This lemma implies that many results which hold for $\mathcal{F}$ also hold for $\mathcal{F}_{\mathcal{U}|s}$. For what follows, we need Lemma \ref{lem:inf} and Proposition \ref{prop:l1bound1} that are given next. Lemma \ref{lem:inf} gives tight bounds on $P_x$ given $\hat{P}$ and $B$. Will use this in Section \ref{sec:poly} to describe polyhedral approximations of $\mathcal{F}$ and the $\mathcal{F}_s$, which we will use in turn to obtain useful lower bounds on the utility that can be obtained under RLDP.

\begin{lemma} \label{lem:inf}
Let $x \in \mathcal{X}$. Then
\begin{align}
\min_{P \in \mathcal{F}} P_x &= \frac{B+2\hat{P}_x-\sqrt{B^2+4B\hat{P}_x-4B\hat{P}_x^2}}{2B+2},\\
\max_{P \in \mathcal{F}} P_x &= \frac{B+2\hat{P}_x+\sqrt{B^2+4B\hat{P}_x-4B\hat{P}_x^2}}{2B+2}.
\end{align}
\end{lemma}

\begin{proof}
Evidently the minimum and maximum exist and are attained on the boundary, i.e. for $P$ satisfying $\sum_{x'} \frac{\hat{P}_{x'}^2}{P_{x'}} = B+1$. Thus for finding both the minimum and the maximum we have to find the stationary points of
\begin{equation}
P_x + \lambda\left(\sum_{x'} \frac{\hat{P}_{x'}^2}{P_{x'}}-B-1\right)+\mu\left(\sum_{x'} P_{x'}-1\right).
\end{equation}
Taking derivatives with respect to all $P_{x'}$, we find
\begin{align}
1+\mu -\lambda \frac{\hat{P}_x^2}{P_x^2} &= 0,\\
\forall x'\neq x: \ \mu-\lambda\frac{\hat{P}_{x'}^2}{P_{x'}^2} &= 0.
\end{align}
It follows that for $x' \neq x$, we have $P_{x'} = c\hat{P}_{x'}$, with $c = \sqrt{\frac{\lambda}{\mu}}$. Since $\sum_{x'} P_{x'} = \sum_{x'} \hat{P}_{x'} = 1$, hence $c = \frac{1-P_x}{1-P_{x'}}$. Substituting this in the boundary constraint yields
\begin{equation}
\frac{\hat{P}_x^2}{P_x} + \frac{(1-\hat{P}_x)^2}{1-P_x} = B+1.
\end{equation}
Solving this for $P_x$ gives us
\begin{equation}
P_x = \frac{B+2\hat{P}_x\pm\sqrt{B^2+4B\hat{P}_x-4B\hat{P}_x^2}}{2B+2},
\end{equation}
giving both the minimum and maximum.
\end{proof}

The following Proposition gives a bound on $||P-\hat{P}||_1$ in terms of $\hat{P}$ and $B$, which is tight for $B \geq 1$. This is an essential ingredient to the explicit privacy protocols introduced in Sections \ref{sec:sep} and \ref{sec:cond}. The proof is rather long and technical, so we present it in Appendix \ref{ssec:pfl1bound1}.

\begin{proposition} \label{prop:l1bound1}
Let $B$ and $\hat{P} \in \mathcal{P}_{\mathcal{X}}$ be given. 
\begin{enumerate}
    \item Suppose $B \geq 1$. Let $x_{\recht{min}} \in \argmin_{x\in\mathcal{X}} \hat{P}_x$. Then 
\begin{equation}
\max_{P \in \mathcal{F}} ||P-\hat{P}||_1 = \frac{B-2B\hat{P}_{x_{\recht{min}}}+\sqrt{B^2+4B\hat{P}_{x_{\recht{min}}}-4B\hat{P}_{x_{\recht{min}}}^2}}{B+1}. \label{eq:l1sup1}
\end{equation}
    \item Suppose $B < 1$. Then $\max_{P \in \mathcal{F}} ||P-\hat{P}||_1 \leq \sqrt{B}.$
\end{enumerate}
\end{proposition}

\section{Polyhedral approximation: PolyOpt} \label{sec:poly}

Our first method to find RLDP protocols for when $\mathcal{F}$ is a confidence interval from a $\chi^2$ test relies on optimising $\recht{I}_{P}(X;Y)$ over protocols that satisfy a more stringent privacy constraint; this yields a lower bound on the maximal $\recht{I}_{P}(X;Y)$. More concretely, we consider protocols that satisfy (\ref{eq:rpp}) for all $P$ for which $P_{\mathcal{U}|s} \in \mathcal{D}_{\mathcal{U}|s}$, where each $\mathcal{D}_{\mathcal{U}|s}$ is a polyhedron containing the set $\mathcal{F}_{\mathcal{U}|s}$ from Lemma \ref{lem:bs}. All $P \in \mathcal{F}$ certainly satisfy this condition. For  each $s,u$, let $P^{\recht{min}}_{u|s} = \inf_{P \in \mathcal{F}} P_{u|s}$: an explicit formula is given in Lemma \ref{lem:inf}.
When each $\mathcal{D}_{\mathcal{U}|s}$ is the simplex $\{R: \forall u \ R_u \geq P_{u|s}^{\recht{min}}\}$, robust optimisation for polytopes \cite{ben2009robust} yields the following result. Let $\Gamma$ be the convex cone consisting of all $T \in \mathbb{R}^{\mathcal{X}}_{\geq 0}$ satisfying 
\begin{equation}
\forall s_1,s_2,u_1,u_2: T_{s_1,u_1}-\textrm{e}^{\varepsilon}T_{s_2,u_2}+\sum_uP^{\recht{min}}_{u|s_1}\left(T_{s_1,u}-T_{s_1,u_1}\right) - \textrm{e}^{\varepsilon}\sum_uP^{\recht{min}}_{u|s_2}\left(T_{s_2,u}-T_{s_2,u_2}\right) \leq 0.
\end{equation}

\begin{theorem} \label{thm:poly}
Let $\mathcal{Q}$ be a privacy protocol such that for all $y$ we have
$Q_y \in \Gamma$. Then $\mathcal{Q}$ satisfies $(\varepsilon,\mathcal{F})$-RLDP.
\end{theorem}

\begin{theorem} \label{thm:poly2}
Let $\hat{\Gamma}$ be polytope given by $\{T \in \Gamma: \sum_x T_x = 1\}$. Let $\mathcal{V}$ be the set of vertices of $\hat{\Gamma}$. For $v \in \mathcal{V}$, define
\begin{align}
    \mu^1(v) &= \sum_x v_x \hat{P}_x \log \frac{v_x}{\sum_{x'} v_{x'}\hat{P}_{x'}},\\
    \mu^2(v) &= \min_{P \in \mathcal{F}} \sum_x v_x P_x \log \frac{v_x}{\sum_{x'} v_{x'}P_{x'}}.
\end{align}
For $i = 1,2$, let $\hat{\theta}^i$ be the solution to the optimisation problem
    \begin{align}
    \recht{maximise}_{\theta} & \ \sum_{v \in \mathcal{V}} \theta_v \mu^i(v) \label{eq:polyprob}\\
    \recht{satisfying} & \ \theta \in \mathbb{R}^{\mathcal{V}}_{\geq 0}, \nonumber\\
    & \ \sum_v \theta_vv = 1_{\mathcal{X}}. \nonumber
    \end{align}
Let the privacy protocol $\mathcal{Q}^i$ be given by $\mathcal{Y}^i = \{v \in \mathcal{V}: \hat{\theta}^i_v > 0\}$ and $Q^i_{v|x} = \hat{\theta}^i_v v_x$. Then:
\begin{enumerate}
    \item The protocol $\mathcal{Q}^1$ maximises $\recht{I}_{\hat{P}}(X;Y)$ among all protocols satisfying the condition of Theorem \ref{thm:poly}. One has $\left|\mathcal{Y}^i\right| \leq a$.
    \item Let $L = \sum_{v \in \mathcal{V}} \hat{\theta}^2_v \mu^2(v)$. Then $\mathcal{Q}^2$ satisfies $\inf_{P \in \mathcal{F}} \recht{I}_{P}(X;Y) \geq L$.
\end{enumerate}
\end{theorem}

Together, these two theorems show, if we can solve a vertex enumeration problem, that we can find a protocol $\mathcal{Q}^1$ that maximises $\recht{I}_{\hat{P}}(X;Y)$ among a subset of all $(\varepsilon,\mathcal{F})$-RLDP $\mathcal{Q}$, a lower bound for the achievable $\min_P \recht{I}_{P}(X;Y)$, and a protocol $\mathcal{Q}^2$ that exceeds this bound.

In Theorem \ref{thm:poly2}, to calculate $\mu^2(v)$ one needs to take the minimum over all $P \in \mathcal{F}$. To approximate this, one may replace $\mathcal{F}$ by a polyhedron containing it; the minimum is then attained at one of its vertices.

Before we prove Theorem \ref{thm:poly}, we need an intermediate result. For a privacy protocol $\mathcal{Q}$ and a $y \in \mathcal{Y}$, we let $Q_y$ be the vector $(Q_{y|x})_x \in \mathbb{R}^{\mathcal{X}}$. Furthermore, for $s_1,s_2 \in \mathcal{S}$, let $B^{s_1,s_2} \in \mathbb{R}^{\mathcal{X} \times \mathcal{X}}$ be given by
\begin{equation}
B^{s_1,s_2}_{s,u;s',u'} = \left\{\begin{array}{ll}
1, & \textrm{ if $u = u'$ and $s = s' = s_1$,}\\
-\textrm{e}^{\varepsilon}, & \textrm{ if $u = u'$ and $s = s' = s_2$,}\\
0, & \textrm{ otherwise.}
\end{array}\right.
\end{equation}

\begin{lemma} \label{lem:poly}
Let $\mathcal{D} \subset (\mathcal{P}_{\mathcal{U}})^{\mathcal{S}} \subset \mathbb{R}^{\mathcal{X}}$ be a polyhedron such that for every $P \in \mathcal{F}$ one has $(P_{\mathcal{U}|s})_{s \in \mathcal{S}} \in \mathcal{D}$. Let $\mathcal{D}$ be given by the equations $DR+d \geq 0$ and $ER+e = 0$, for matrices $D$ and $E$, vectors $d$ and $e$, and $R \in \mathbb{R}^{\mathcal{S} \times \mathcal{U}}$. Let $\mathcal{Q}$ be a privacy protocol such that for all $y \in \mathcal{Y}$ and $s_1,s_2 \in \mathcal{S}$ there exist $z,w$ such that
\begin{align}
D^{\recht{T}}z+E^{\recht{T}}w &= - (B^{s_1,s_2})^{\recht{T}}Q_y, \\
z & \geq 0, \label{eq:poly1}\\
d^{\recht{T}}z+e^{\recht{T}}w &\leq 0. \label{eq:poly2}
\end{align}
Then $\mathcal{Q}$ satisfies $\varepsilon$-RLDP w.r.t. $\mathcal{F}$.
\end{lemma}

\begin{proof}
For $y \in \mathcal{Y}$ and $s \in \mathcal{S}$, write $Q_{y,s} := (Q_{y|s,u})_u \in \mathbb{R}^{\mathcal{U}}$, and $Q_y := (Q_{y|s,u})_{s,u} \in \mathbb{R}^{\mathcal{X}}$. We can then formulate $\varepsilon$-RLDP as
\begin{equation} \label{eq:polysldp}
\forall y,s_1,s_2: \max_{P \in \mathcal{F}} P_{\mathcal{U}|s_1}^{\recht{T}}Q_{y,s_1} - \textrm{e}^{\varepsilon} P_{\mathcal{U}|s_2}^{\recht{T}}Q_{y,s_2} \leq 0.
\end{equation}
Set $\mathcal{G} = \prod_s \mathcal{F}_{\mathcal{U}|s}$. Then $\mathcal{D}$ satisfies the conditions of the Lemma if and only if $\mathcal{G} \subset \mathcal{D}$. In particular, the following condition implies (\ref{eq:polysldp}):
\begin{equation} \label{eq:polysldp2}
\forall y,s_1,s_2: \max_{R \in \mathcal{D}} R_{s_1}^{\recht{T}}Q_{y,s_1} - \textrm{e}^{\varepsilon} R_{s_2}^{\recht{T}}Q_{y,s_2} \leq 0.
\end{equation}
Using the matrices $B^{s_1,s_2}$, we can rewrite (\ref{eq:polysldp2}) as 
\begin{equation} \label{eq:polysldp3}
\forall y,s_1,s_2: \max_{R \in \mathcal{D}} ((B^{s_1,s_2})^{\recht{T}}Q_y)^{\recht{T}}R \leq 0.
\end{equation}
Now fix $s_1,s_2,y$. By dualising we have
\begin{equation}
\max_{R \in \mathcal{D}} ((B^{s_1,s_2})^{\recht{T}}Q_y)^{\recht{T}}R = \min_{\substack{z,w:\\D^{\recht{T}}z+E^{\recht{T}}w = -(B^{s_1,s_2})^{\recht{T}}Q,\\z \geq 0}}  d^{\recht{T}}z+e^{\recht{T}}w.
\end{equation}
It follows that $Q$ satisfies $\varepsilon$-RLDP if for each $y,s_1,s_2$ there exist $z\geq 0$ and $w$ satisfying $D^{\recht{T}}z+E^{\recht{T}}w = -(B^{s_1,s_2})^{\recht{T}}Q$ such that $d^{\recht{T}}z+e^{\recht{T}}w \leq 0$.
\end{proof}

\begin{proof}[Proof of Theorem \ref{thm:poly}] 
Define $\mathcal{D}_{\mathcal{U}|s} = \{R \in \mathcal{P}_{\mathcal{S}}: \forall u \ R_u \geq P_{u|s}^{\recht{min}}\}$, and let $\mathcal{D} = \prod_s \mathcal{D}_{\mathcal{U}|s}$. This satisfies the conditions of Lemma \ref{lem:poly}. One checks that in this case we have $D = \recht{id}_{\mathcal{X}}$, $d \in \mathcal{R}^{\mathcal{X}}$ is given by $d_{s,u} = -P^{\recht{min}}_{u|s}$, $E \in \mathbb{R}^{\mathcal{S}\times\mathcal{X}}$ is given by $E_{s;u',s'} = \delta_{s=s'}$, and $e = -1_{\mathcal{S}}$. This also means that $z \in \mathbb{R}^{\mathcal{X}}$ and $w \in \mathbb{R}^{\mathcal{S}}$. It follows from these descriptions that
\begin{align}
D^{\recht{T}}z &= z,\\
(E^{\recht{T}}w)_{s,u} &= w_s,\\
((B^{s_1,s_2})^{\recht{T}}Q_y)_{s,u} &= \left\{\begin{array}{ll}
Q_{y|s_1,u}, & \textrm{ if $s = s_1$,}\\
-\textrm{e}^{\varepsilon}Q_{y|s_2,u}, & \textrm{ if $s = s_2$,}\\
0, & \textrm{ otherwise.}
\end{array}\right.
\end{align}
It follows that $D^{\recht{T}}z + E^{\recht{T}}w = -(B^{s_1,s_2})^{\recht{T}}Q_y$ can be rewritten as
\begin{equation}
    z_{s,u} = \left\{\begin{array}{ll}
    -Q_{y|s_1,u}-w_{s_1}, & \textrm{ if $s=s_1$,}\\
    \textrm{e}^{\varepsilon}Q_{y|s_2,u}-w_{s_2}, & \textrm{ if $s=s_2$,}\\
    -w_s & \textrm{ otherwise.}
    \end{array}\right.
\end{equation}
Eliminating $z$ from (\ref{eq:poly1}) and(\ref{eq:poly2}), we get
\begin{align}
-\sum_s\left(1-\sum_uP^{\recht{min}}_{u|s}\right)w_s+\sum_u Q_{y|s_1,u}P^{\recht{min}}_{u|s_1} - \textrm{e}^{\varepsilon}\sum_u Q_{y|s_2,u}P^{\recht{min}}_{u|s_2} &\leq 0, \label{eq:poly3}\\
\forall u: \ w_{s_1} &\leq -Q_{y|s_1,u}, \label{eq:poly4}\\
\forall u: \ w_{s_2} &\leq \textrm{e}^{\varepsilon}Q_{y|s_2,u}, \label{eq:poly5}\\
\forall s \neq s_1,s_2: \ w_s &\leq 0. \label{eq:poly6}
\end{align}
Since $\sum_u P^{\recht{min}}_{u|s} \leq 1$ for all $s$, it follows that the left hand side of (\ref{eq:poly3}) is minimal if each $w_s$ attains its maximal value, subject to the constraints (\ref{eq:poly4}--\ref{eq:poly6}). It follows that the minimum of the left hand side is equal to
\begin{align}
&\left(1-\sum_u P^{\recht{min}}_{u|s_1}\right)\left(\max_{u_1} Q_{y|u_1,s_1}\right) - \textrm{e}^{\varepsilon} \left(1-\sum_u P^{\recht{min}}_{u|s_2}\right)\left(\min_{u_2} Q_{y|u_2,s_2}\right)\\
&+ \sum_u Q_{y|s_1,u}P^{\recht{min}}_{u|s_1} -\textrm{e}^{\varepsilon}\sum_u Q_{y|s_2,u}P^{\recht{min}}_{u|s_2}. \nonumber
\end{align}
This is nonpositive if and only if it is nonpositive for all choices of $u_1$ and $u_2$; but this is true precisely if $Q_y \in \Gamma$.
\end{proof}

The proof of Theorem \ref{thm:poly2} is analogous to the proof of Theorem 4 of \cite{kairouz2014extremal}. It is presented in Appendix \ref{ssec:pfpoly}. The algorithm that produces $\mathcal{Q}^1$ from $\hat{P}$ and $\varepsilon$ will be refered to as \emph{PolyOpt} in the remainder of this paper.

\begin{remark}
A simplex is not the only possible choice for $\mathcal{D}_{\mathcal{U}|s}$. In general, we can make $\mathcal{D}_{\mathcal{U}|s}$ closer to $\mathcal{F}_{\mathcal{U}|s}$ by adding more defining hyperplanes. Doing this allows more $\mathcal{Q}$ to satisfy Theorem \ref{thm:poly}, and in turn increase the utility of the $\mathcal{Q}$ we find via Theorem \ref{thm:poly2}. However, since $\Gamma$ is related to the $\mathcal{D}_{\mathcal{U}|s}$ via duality, adding extra constraints to the $\mathcal{D}_{\mathcal{U}|s}$ will increase the dimension of $\Gamma$ through the addition of auxiliary variables. This makes the vertex enumeration problem of Theorem \ref{thm:poly2} more computationally involved. Thus we have a tradeoff between utility and computational complexity.

It should be noted that in general the increasing utility found in this way does not approach the optimal utility over all $(\varepsilon,\mathcal{F})$-RLDP protocols. This is because, as we take increasingly finer $\mathcal{D}_{\mathcal{U}|s}$, we approach the set of $\mathcal{Q}$ that satisfy (\ref{eq:rpp}) for all $P$ in $\mathcal{F}' := \{P: \forall s \ P_{\mathcal{U}|s} \in \mathcal{F}_{\mathcal{U}|s}\}$. Since in general $\mathcal{F} \subsetneq \mathcal{F}'$, the set of $(\varepsilon,\mathcal{F}')$-RLDP protocols is strictly smaller than the set of $(\varepsilon,\mathcal{F})$-RLDP protocols.
\end{remark}
% !TEX root =  main.tex

\section{Independent reporting} \label{sec:sep}

As PolyOpt relies on vertex enumeration, it can be computationally infeasible for larger $a$. In this section, we consider a class of release protocols which we call \emph{Independent Reporting}. We show that within this class the optimal protocols can be found by finding the maximum of a one-dimensional function. Since the dimension of this optimisation problem does not depend on $a$, this approach can be used when vertex enumeration is out of reach. As mentioned before we continue to let $\mathcal{F}$ be a confidence set for a $\chi^2$ test.

The basis of IR is to apply two separate LDP protocols $\mathcal{R}^1$ and $\mathcal{R}^2$ to $S$ and $U$, respectively, and output $(\mathcal{R}^1(S),\mathcal{R}^2(U))$. This is described in Protocol \ref{alg:ir}. 

\begin{algorithm}
\SetAlgoLined
\SetKwInOut{Input}{Input}\SetKwInOut{Output}{Output}
\SetKwFunction{Sort}{Sort}
\Input{Privacy protocols $\mathcal{R}^1\colon \mathcal{S} \rightarrow \mathcal{Y}^1$ and $\mathcal{R}^2\colon\mathcal{U} \rightarrow \mathcal{Y}^2$; $x = (s,u) \in \mathcal{X}$.}
\Output{Output datum $y \in \mathcal{Y} := \mathcal{Y}^1 \times \mathcal{Y}^2$}
\BlankLine
Compute $y_1 \leftarrow \mathcal{Q}^1(s)$\;
Compute $y_2 \leftarrow \mathcal{Q}^2(u)$\;
$y \leftarrow (y_1,y_2)$\;

 \caption{$\recht{IR}_{\mathcal{Q}^1,\mathcal{Q}^2}$ (Independent reporting) \label{alg:ir} }
 
\end{algorithm}

While only $S$ needs to be protected, we also need to apply a privacy protocol to $U$ because of the possible correlation between the two. However, since $U$ only indirectly leaks information about $S$, we can get away with less strict privacy requirements. This is reflected in the following theorem.

\begin{theorem} \label{thm:irpriv}
Let $\varepsilon_1,\varepsilon_2 \in \mathbb{R}_{\geq 0}$. For each $s$, define $B_s$ as in Lemma \ref{lem:bs}, and let $u_s \in \mathcal{U}$ be such that $\hat{P}_{u_s|s}$ is minimal. Define
\begin{equation} \label{eq:irpriv}
d_s := \left\{\begin{array}{ll}
\frac{B_s(1-2\hat{P}_{u_s|s})+\sqrt{B_s^2+4B_s\hat{P}_{u_s|s}-4B_s\hat{P}_{u_s|s}^2}}{B_s+1}, & \textrm{if $B_s \geq 1$;}\\
\sqrt{B_s}, & \textrm{if $B_s \leq 1$.}
\end{array}
\right.
\end{equation}
Furthermore, define
\begin{equation}
d := \recht{min}\left\{2,\max_s(2d_s)+\max_{s,s'}||\hat{P}_{\mathcal{U}|s}-\hat{P}_{\mathcal{U}|s'}||_1\right\}.
\end{equation}
Let $\delta_2 = \log\left(1+\frac{2(\textrm{e}^{\varepsilon_2}-1)}{d}\right)$. Suppose that $\mathcal{R}^1$ is $\varepsilon_1$-LDP and that $\mathcal{R}^2$ is $\delta_2$-LDP. Then IR is $(\varepsilon_1+\varepsilon_2,\mathcal{F})$-RLDP.
\end{theorem}

\begin{proof}
We start by showing that $d$ is an upper bound for $||P_{\mathcal{U}|s}-\hat{P}_{\mathcal{U}|s'}||_1$. If $d = 2$, this is certainly the case. Suppose $d = \max_s(2d_s)+\max_{s,s'}||\hat{P}_{\mathcal{U}|s}-\hat{P}_{\mathcal{U}|s'}||_1$. It follows from Lemma \ref{prop:l1bound1} that for each $P \in \mathcal{F}$ and each $s \in \mathcal{S}$ we have $||P_{\mathcal{U}|s}-\hat{P}_{\mathcal{U}|s}||_1 \leq d_s$. Hence, for all $s,s'\in\mathcal{S}$ and $P\in\mathcal{F}$ we have
\begin{align}
||P_{\mathcal{U}|s}-\hat{P}_{\mathcal{U}|s'}||_1 &\leq ||P_{\mathcal{U}|s}-\hat{P}_{\mathcal{U}|s}||_1 + ||\hat{P}_{\mathcal{U}|s}-\hat{P}_{\mathcal{U}|s'}||_1 + ||\hat{P}_{\mathcal{U}|s'}-P_{\mathcal{U}|s'}||_1 \\
&\leq d_s+d_{s'} + ||\hat{P}_{\mathcal{U}|s}-\hat{P}_{\mathcal{U}|s'}||_1 \\
&\leq d.
\end{align}
Combining Lemma \ref{lem:ldp} with the fact that $\varepsilon_2 = \log\left(1+\frac{d(\textrm{e}^{\delta_2}-1)}{2}\right)$, it follows that for every $y_2 \in \mathcal{Y}_2$ we have
\begin{align}
\frac{\mathbb{P}_P(\mathcal{R}^2(U) = y_2 | S = s)}{\mathbb{P}_P(\mathcal{R}^2(U) = y_2 | S = s')} &\leq 1+\frac{\textrm{e}^{\delta_2}-1}{2}||P_{\mathcal{U}|s}-P_{\mathcal{U}|s'}||_1 \\
&\leq 1+\frac{d(\textrm{e}^{\delta_2}-1)}{2}\\
&= \textrm{e}^{\varepsilon_2}.
\end{align}
Since $\mathcal{R}^1$ is $\varepsilon_1$-LDP, it follows that for every $y_1 \in \mathcal{Y}_1$ and every $y_2 \in \mathcal{Y}_2$ we have
\begin{equation}
\frac{\mathbb{P}(\mathcal{R}^1(S) = y_1, \mathcal{R}^2(U) = y_2 | S = s)}{\mathbb{P}(\mathcal{R}^1(S) = y_1, \mathcal{R}^2(U) = y_2 | S = s')} \leq \textrm{e}^{\varepsilon_1+\varepsilon_2},
\end{equation}
which shows that $\recht{IR}_{\mathcal{R}^1,\mathcal{R}^2}$ is $(\varepsilon_1+\varepsilon_2, \mathcal{F})$-RLDP.
\end{proof}

The more independent $S$ and $U$ are, the smaller $\max_{s,s'}||\hat{P}_{\mathcal{U}|s}-\hat{P}_{\mathcal{U}|s'}||_1$ will be. Theorem \ref{thm:irpriv} then tells us that for more independent $S$ and $U$, the privacy requirements on $\mathcal{R}^2$ will be less strict, resulting in better utility. The utility of IR is described by the following theorem:

\begin{theorem} \label{thm:iruti}
For any $P\in\mathcal{P}_\mathcal{X}$ one has
\begin{equation}
\recht{I}_P(\recht{IR}_{\mathcal{R}^1,\mathcal{R}^2}(X);X)  = \recht{I}_P(\mathcal{R}^1(S);S) + \recht{I}_P(\mathcal{R}^2(U);U|\mathcal{R}^1(S)).
\end{equation}
\end{theorem}

\begin{proof}
Since $\mathcal{R}^1(S)$ and $U$ are independent given $S$, and $\mathcal{R}^2(U)$ and $S$ are independent given $U$, we have
\begin{align}
\recht{I}_P(\recht{IR}_{\mathcal{R}^1,\mathcal{R}^2}(X);X) &= \recht{I}_P(\mathcal{R}^1(S),\mathcal{R}^2(U);U,S)\\
&= \recht{I}_P(\mathcal{R}^1(S);U,S) + \recht{I}_P(\mathcal{R}^2(U);U,S|\mathcal{R}^1(S))\\
&= \recht{I}_P(\mathcal{R}^1(S);S) + \recht{I}_P(\mathcal{R}^2(U);U|\mathcal{R}^1(S)). \qedhere
\end{align}
\end{proof}

Given an $\varepsilon \geq 0$, we can use these theorems to find $(\varepsilon,\mathcal{F})$-RLDP protocols. Per Theorem \ref{thm:irpriv}, it suffices to take a $\varepsilon_2$, and use a $\varepsilon_1$-LDP protocol $\mathcal{R}^1$ and a $\delta_2$-LDP protocol $\mathcal{R}^2$, where $\varepsilon_1 = \varepsilon - \varepsilon_2$ and $\delta_2$ is as in Theorem \ref{thm:irpriv}. We want to choose $\varepsilon_2$, $\mathcal{R}^1$ and $\mathcal{Q}^2$ in such a way that we maximise the expression in Theorem \ref{thm:iruti}. For $\varepsilon$ large enough, the $\mathcal{R}^1$ that maximises $\recht{I}_P(\mathcal{R}^1(S);S)$ is GRR. Furthermore, since
\begin{equation}
\recht{I}_P(\mathcal{R}^2(U);U|\mathcal{R}^1(S)) = \mathbb{E}_{r}\left[\recht{I}_{U \sim P_{U}|\mathcal{R}^1(S)=r}(\mathcal{R}^2(U);U)\right],
\end{equation}
and GRR maximises $\recht{I}(\mathcal{R}^2(U);U)$ for large enough $\varepsilon$ for any distribution of $U$, we should take $\mathcal{R}^2$ to be GRR as well; this is true regardless of the value of $P$. We are left with only the unknown $\varepsilon_2$, hence to maximise the mutual information of IR for a given $P$ we have to solve a one-dimensional optimisation problem.
\section{Conditional reporting} \label{sec:cond}

From Theorem \ref{thm:irpriv} it is clear that in IR we can afford a larger privacy budget to $\mathcal{R}^2$ if $S$ and $U$ are only weakly correlated. When $S$ and $R$ are closer related, however, the difference between $\delta_2$ and $\varepsilon_2$ will be small, and IR cannot offer any advantage over general LDP protocols. To this end, we introduce two other protocols that fall under the umbrella term \emph{Conditional Reporting}. In both these protocols, we apply an established privacy protocol $\mathcal{R}^1$ to $S$. Furthermore, we return $U$ (with a small perturbation) if $\mathcal{R}^1$ returns a `correct' response and a random $U$ otherwise. We will see that the noise on $U$ depends on the size of the feasible set $\mathcal{F}$ rather than on the correlation between $S$ and $U$.

\subsection{GRR-CR}

For the first CR protocol, GRR-CR, we first need to specify a parameter $\varepsilon_1$ and, for each $s \in \mathcal{S}$, a privacy protocol $\mathcal{R}^s\colon \mathcal{U} \rightarrow \mathcal{Y}_s$, where each $\mathcal{Y}_s$ is a finite set. To apply it to an input datum $(s,u) \in \mathcal{X}$, we first apply GRR with parameter $\varepsilon_1$ to $s$; call the outcome $\tilde{S}$. If $\tilde{S} = s$, we apply $\mathcal{R}^s$ to $u$, and we output $(s,\mathcal{R}^s(u))$. If $\tilde{S} \neq s$, we draw a random $\tilde{U} \in \mathcal{U}$ from the probability distribution $\hat{P}_{\mathcal{U}|\tilde{S}}$, and we output $(\tilde{S},\mathcal{R}^{\tilde{S}}(\tilde{U}))$. This protocol is described in Protocol \ref{alg:grr-cr}.

\begin{algorithm}
\SetAlgoLined
\SetKwInOut{Input}{Input}\SetKwInOut{Output}{Output}
\SetKwFunction{Sort}{Sort}
\Input{Privacy parameter $\varepsilon_1$; For every $s \in \mathcal{S}$, a privacy protocol $\mathcal{Q}^s\colon\mathcal{U} \rightarrow \mathcal{Y}_s$; input datum $x = (s,u) \in \mathcal{X}$}
\Output{Output datum $Y \in \mathcal{S} \times \bigcup_{s \in \mathcal{S}} \mathcal{Y}_s$}
\BlankLine
Take $\tilde{S} \leftarrow \recht{GRR}^{\varepsilon_1}(s) \in \mathcal{S}$\;
 \caption{GRR-CR}
\eIf{$\tilde{S} = s$}{Compute $Y \leftarrow (s,\mathcal{R}^s(u))$\;}{Sample $\tilde{U} \in \mathcal{U}$ with $\mathbb{P}(\tilde{U} = u') = \hat{P}_{u'|\tilde{S}}$\;
Compute $Y \leftarrow (\tilde{S},\mathcal{R}^{\tilde{S}}(\tilde{U}))$\;}
Output $Y$\;
 \label{alg:grr-cr}
\end{algorithm}

Although we have already obfuscated $S$ via GRR, we still need to obfuscate $U$ and $\tilde{U}$ via $\mathcal{R}^{\tilde{S}}$ for the following reason. Suppose we omit this last step, and instead return $(\tilde{S},\tilde{U})$, with $\tilde{U} = u$ if $\tilde{S} = s$. From the viewpoint of an attacker, given $\tilde{S}$, the random variable $\tilde{U}$ is drawn from the distribution $P_{\mathcal{U}|\tilde{S}}$ if $\tilde{S} = s$, and from the distribution $\hat{P}_{\mathcal{U}|\tilde{S}}$ otherwise. In the LDP model the attacker may collude with an arbitrary amount of users, and as such we may assume that they have access to the real distribution $P \in \mathcal{P}_{\mathcal{X}}$. Under this assumption, the output $\tilde{U}$ contains information about whether it was drawn from $P_{\mathcal{U}|\tilde{S}}$ or $\hat{P}_{\mathcal{U}|\tilde{S}}$, and hence whether $S = \tilde{S}$ or not. To prevent this leakage, we have to mask $\tilde{U}$ with the privacy protocol $\mathcal{R}^{\tilde{S}}$. As the following theorem shows, the privacy level that is needed for $\mathcal{R}^{s}$ depends on $||\hat{P}_{\mathcal{U}|s}-P_{\mathcal{U}|s}||_1$, which explains why we need a different protocol $\mathcal{R}^s$ for every $s$.

\begin{theorem} \label{thm:grrcrpriv}
Let $\varepsilon_1, \varepsilon_2 \in \mathbb{R}_{\geq 0}$. For every  $s \in \mathcal{S}$ define $d_s$ as in (\ref{eq:irpriv}), define $\delta_s := \log\left(1+\frac{2(\textrm{e}^{\varepsilon_2}-1)}{d_s}\right)$, and let $\mathcal{Q}^s$ satisfy $\delta_s$-LDP. Then Algorithm \ref{alg:grr-cr} satisfies $(\varepsilon_1+\varepsilon_2, \mathcal{F})$-RLDP.
\end{theorem}

\begin{proof}
For $s \in \mathcal{S}$, $u \in \mathcal{U}$, and $y \in \mathcal{Y}_s$, let $R^s_{y|u} = \mathbb{P}(\mathcal{Q}^S(u) = y)$. Then for every $\tilde{s},s \in \mathcal{S}$ en every $y \in \mathcal{Y}_{\tilde{s}}$ we have
\begin{equation}
\mathbb{P}(\tilde{S} = \tilde{s},\mathcal{R}^{\tilde{s}}(U) = y |S = s) = \left\{
\begin{array}{ll}
\frac{\textrm{e}^{\varepsilon_1}\sum_u R^{\tilde{s}}_{y|u}P_{u|\tilde{s}}}{\textrm{e}^{\varepsilon_1}+a_1-1},& \textrm{if $s = \tilde{s}$,}\\
\frac{\sum_u R^{\tilde{s}}_{y|u}\hat{P}_{u|\tilde{s}}}{\textrm{e}^{\varepsilon_1}+a_1-1},& \textrm{if $s \neq \tilde{s}$.}
\end{array}
\right.
\end{equation}
It follows that for every $s' \in \mathcal{S}$ we have
\begin{align}
\frac{\mathbb{P}(\tilde{S} = \tilde{s},\mathcal{R}^{\tilde{s}}(U) = y |S = s)}{\mathbb{P}(\tilde{S} = \tilde{s},\mathcal{R}^{\tilde{s}}(U) = y |S = s')} &= \left\{
\begin{array}{ll}
1,& \textrm{if $s=s'$}\\
\textrm{e}^{\varepsilon_1}\frac{\sum_u R^{\tilde{s}}_{y|u}P_{u|\tilde{s}}}{\sum_u R^{\tilde{s}}_{y|u}\hat{P}_{u|\tilde{s}}},& \textrm{if $s = \tilde{s} \neq s'$}\\
\textrm{e}^{-\varepsilon_1}\frac{\sum_u R^{\tilde{s}}_{y|u}\hat{P}_{u|\tilde{s}}}{\sum_u R^{\tilde{s}}_{y|u}P_{u|\tilde{s}}},& \textrm{if $s \neq \tilde{s} = s'$}\\
1,& \textrm{if $s \neq \tilde{s} \neq s'$}
\end{array}
\right.\\
&\leq \textrm{e}^{\varepsilon_1}\recht{max}\left\{\frac{\sum_u R^{\tilde{s}}_{y|u}P_{u|\tilde{s}}}{\sum_u R^{\tilde{s}}_{y|u}\hat{P}_{u|\tilde{s}}},\frac{\sum_u R^{\tilde{s}}_{y|u}\hat{P}_{u|\tilde{s}}}{\sum_u R^{\tilde{s}}_{y|u}P_{u|\tilde{s}}}\right\}.
\end{align}
Since $||P_{\mathcal{U}|s}-\hat{P}_{\mathcal{U}|s}||_1 \leq d_s$, we find by Lemma \ref{lem:ldp} that
\begin{align}
\frac{\sum_u R^{\tilde{s}}_{y|u}P_{u|\tilde{s}}}{\sum_u R^{\tilde{s}}_{y|u}\hat{P}_{u|\tilde{s}}} &= \frac{\mathbb{P}_{\tilde{U} \sim P_{\mathcal{U}|\tilde{s}}}(\tilde{Q}^{\tilde{s}}(\tilde{U}) = y)}{\mathbb{P}_{\tilde{U} \sim \hat{P}_{\mathcal{U}|\tilde{s}}}(\tilde{Q}^{\tilde{s}}(\tilde{U}) = y)} \\
&\leq 1+\frac{||P_{\mathcal{U}|s}-\hat{P}_{\mathcal{U}|s}||_1(\textrm{e}^{\delta_s}-1)}{2}\\
&\leq 1+\frac{d_s(\textrm{e}^{\delta_s}-1)}{2}\\
&= \textrm{e}^{\varepsilon_2}.
\end{align}
The same holds analogously for $\frac{\sum_u R^{\tilde{s}}_{y|u}\hat{P}_{u|\tilde{s}}}{\sum_u R^{\tilde{s}}_{y|u}P_{u|\tilde{s}}}$, and it follows that $\recht{GRR-CR}$ satisfies $(\varepsilon_1+\varepsilon_2,\mathcal{F})$-RLDP w.r.t. $\mathcal{F}$.
\end{proof}

As we can see, the privacy level of $\mathcal{R}^s$ only depends on $\max_P ||P_{\mathcal{U}|s}-\hat{P}_{\mathcal{U}|s}||_1$. This makes GRR-CR an attractive protocol if this is small, which happens if either the number of known data points $n$ is large or if $\alpha$ is small.

On the side of utility, we have the following:

\begin{theorem} \label{thm:grrcruti}
For any $P$ one has
\begin{equation}
\recht{I}_P(\recht{GRR-CR}(X);X) = \recht{I}_P(\recht{GRR}^{\varepsilon_1}(S);S) + \frac{\textrm{e}^{\varepsilon_1}}{\textrm{e}^{\varepsilon_1}+a_1-1}\recht{I}_P(\mathcal{R}^S(U);U|S).
\end{equation}
\end{theorem}

\begin{proof}
One has
\begin{align}
\recht{I}_P(\recht{GRR-CR}(X);X) &= \recht{I}_P(\tilde{S};S) + \recht{I}_P(\mathcal{R}^{\tilde{S}}(U);S|\tilde{S})\\
& \ \ +\recht{I}_P(\tilde{S};U|S) + \recht{I}_P(\mathcal{R}^{\tilde{S}}(U);U|S,\tilde{S}). \nonumber
\end{align}
Note that $\mathbb{P}(\mathcal{Q}^{\tilde{S}}(U) = y|S= s,\tilde{S} = \tilde{s}) = \sum_u R^{\tilde{s}}_{y|u}\hat{P}_{u|\tilde{s}}$. This does not depend on $s$, hence $S$ and $\mathcal{R}^{\tilde{S}}(U)$ are independent given $\tilde{S}$, and $\recht{I}(\mathcal{R}^{\tilde{S}}(U);S|\tilde{S}) = 0$. Furthermore, $\tilde{S}$ and $U$ are independent given $S$, hence $\recht{I}_P(\tilde{S};U|S) = 0$. For the last term we have
\begin{align}
\recht{I}_P(\mathcal{R}^{\tilde{S}}(U);U|S,\tilde{S}) &= \sum_{s,\tilde{s}} \mathbb{P}(\tilde{S} = \tilde{s}|S =s)\mathbb{P}(S = s)\recht{I}_P(\mathcal{R}^{\tilde{s}}(U);U|S = s,\tilde{S} = \tilde{s}).
\end{align}
We know that $\mathcal{R}^s(U)$ and $U$ are independent given $S$ and $\tilde{S}$ if $S \neq \tilde{S}$, hence in the summation above only terms with $s = \tilde{s}$ matter; hence this is equal to
\begin{align}
\sum_s \mathbb{P}(\tilde{S}=s|S=s)\mathbb{P}(S=s)\recht{I}_P(\mathcal{R}^{s}(U);U|S = \tilde{S} = \tilde{s}) = \frac{\textrm{e}^{\varepsilon_1}}{\textrm{e}^{\varepsilon_1}+a_1-1}\recht{I}_P(\mathcal{R}^S(U);U|S).
\end{align}
The theorem now follows from putting this all together.
\end{proof}

Compared to Theorem \ref{thm:iruti}, we see that if we take $\mathcal{R}^s = \mathcal{R}^2$ for every $s$, then GRR-CR typically has a lower utility than IR. However, the advantage of GRR-CR is that $\mathcal{R}^s$ can typically chosen with more relaxed privacy conditions than $\mathcal{R}^2$, which will increase the utility again.
%We will do experiments to see which method actually delivers the best results in practice.

Since $\recht{I}_P(\mathcal{R}^S(U);U|S) = \sum_s \mathbb{P}(S = s) \recht{I}_{U \sim P_{\mathcal{U}|s}}(\mathcal{R}^s(U);U)$, Theorem \ref{thm:grrcruti} tells us that we want to choose each $\mathcal{R}^s$ to be the $\delta_s$-LDP protocol for which $\recht{I}_{U \sim P_{\mathcal{U}|s}}(\mathcal{R}^s(U);U)$ is maximised. For big enough $\delta_s$, this is GRR. As was the case with IR, it is a one-dimensional optimisation problem to optimise for $\recht{I}_{\hat{P}}(X;Y)$.

\subsection{UE-CR}

The second CR protocol we introduce originates from Unary Encoding (UE) \cite{wang2017locally}. UE is a protocol $\recht{UE}^{\kappa,\lambda}\colon\mathcal{S} \rightarrow 2^{\mathcal{S}}$ given by parameters $0 < \lambda \leq \kappa < 1$ that for an input $s$ outputs a binary vector $(E_{s'})_{s' \in \mathcal{S}}$, where each coefficient is an independent Bernoulli variable with
\begin{equation}
\mathbb{P}(E_{s'} = 1) = \left\{\begin{array}{ll}
\kappa,& \textrm{ if $s = s'$,}\\
\lambda,& \textrm{ if $s \neq s'$.}
\end{array}\right.
\end{equation}
This protocol satisfies $\varepsilon$-LDP for $\varepsilon \geq \log\frac{\kappa(1-\lambda)}{\lambda(1-\kappa)}$. Popular choices for $(\kappa,\lambda)$ are $(\tfrac{\textrm{e}^{\varepsilon/2}}{\textrm{e}^{\varepsilon/2}+1},\tfrac{1}{\textrm{e}^{\varepsilon/2}+1})$, $(\tfrac{1}{2},\tfrac{1}{\textrm{e}^{\varepsilon}+1})$, and $(\tfrac{\textrm{e}^{\varepsilon}}{\textrm{e}^{\varepsilon}+1},\tfrac{1}{2})$ \cite{wang2017locally}. It will be convenient for us to consider the output of UE as a subset of $\mathcal{S}$, rather than a binary vector.

To apply UE-CR to a $(s,u) \in \mathcal{X}$, we first fix parameters $\kappa,\lambda$, and a privacy protocol $\mathcal{R}^s\colon \mathcal{U} \rightarrow \mathcal{Y}^s$ for every $s$. We perform UE on $s$, yielding a subset $\tilde{S} \subset \mathcal{S}$. For every $s' \in \mathcal{S}$, we output a $Y_{s'} \in \mathcal{Y}^{s'}$ as follows: if $s' = s$, we take $Y_{s} = \mathcal{R}^s(u)$. If $s' \neq s$, we draw a $\tilde{U} \in \mathcal{U}$ with probability distribution $\hat{P}_{\mathcal{U}|s'}$, and we take $Y_{s'} = \mathcal{R}^{s'}(\tilde{U})$. Finally, we output $(\tilde{S},(Y_{s'})_{s' \in \tilde{S}})$. This is described in Protocol \ref{alg:ue-cr}.

\begin{algorithm}
\SetAlgoLined
\SetKwInOut{Input}{Input}\SetKwInOut{Output}{Output}
\SetKwFunction{Sort}{Sort}
\Input{Parameters $0 \leq \lambda \leq \kappa \leq 1$; For every $s \in \mathcal{S}$, a privacy protocol $\mathcal{R}^s\colon\mathcal{U} \rightarrow \mathcal{Y}_s$; a probability distribution $\hat{P}$ on $\mathcal{X}$; input datum $x = (s,u) \in \mathcal{X}$}
\Output{Output datum $(\tilde{S},(Y_{s'})_{s' \in \tilde{S}})$ with $\tilde{S} \subset \mathcal{S}$ and $Y_{s'} \in \mathcal{Y}^{s'}$ for each $s'$}
\BlankLine

Take $\tilde{S} \leftarrow \recht{UE}^{\kappa,\lambda}(s) \subset \mathcal{S}$\;

\BlankLine
\For{$s' \in \tilde{S}$}{
\eIf{$s'= s$}{
$Y_s \leftarrow \mathcal{R}^s(u)$\;
}{
Sample $\tilde{U} \sim \hat{P}_{\mathcal{U}|s}$\;
$Y_{s'} \leftarrow \mathcal{R}^s(\tilde{U})$\;}
}
Output $(\tilde{S},(Y_{s'})_{s' \in \tilde{S}})$\;
\caption{UE-CR \label{alg:ue-cr}}
\end{algorithm}

As for GRR-CR, the privacy protocols $\mathcal{R}^{s'}$ are needed to obfuscate the difference between $\hat{P}_{\mathcal{U}|s}$ and $P_{\mathcal{U}|s}$. The precise privacy requirements for the $\mathcal{R}^s$ are given in the theorem below.

\begin{theorem} \label{thm:uepriv}
Let $\varepsilon_1, \varepsilon_2 \in \mathbb{R}_{\geq 0}$. For every $s$, let $\delta_s$ be as in Theorem \ref{thm:irpriv}, and assume $\mathcal{R}^s$ has $\delta_s$-LDP and that $\frac{\kappa(1-\lambda)}{\lambda(1-\kappa)} \leq \textrm{e}^{\varepsilon_1}$. Then UE-CR satisfies $(\recht{max}\{\varepsilon_1+\varepsilon_2,2\varepsilon_2\},\mathcal{F})$-SLDP w.r.t. $\mathcal{F}$.
\end{theorem}

\begin{proof}
For $s \in \mathcal{S}$, define $T_s := \sum_u R^s_{y_s|u}P_{u|s}$ and $\hat{T}_s := \sum_u R^s_{y_s|u}\hat{P}_{u|s}$. One has
\begin{align}
&\mathbb{P}(Y = (\tilde{s},(y_{s'})_{s' \in \tilde{s}})|S = s) \\ 
&= \left\{\begin{array}{ll}
\kappa\lambda^{|\tilde{s}|-1}(1-\lambda)^{a_1-|\tilde{s}|}T_s\prod_{s' \in \tilde{S}\setminus\{s\}} \hat{T}_{s'}, & \textrm{if $s \in \tilde{s}$,}\\
\lambda^{|\tilde{s}|}(1-\kappa)(1-\lambda)^{a_1-|\tilde{s}|-1}\prod_{s' \in \tilde{S}} \hat{T}_{s'}, & \textrm{if $s \notin \tilde{s}$.}
\end{array}
\right.
\end{align}
It follows that
\begin{align}
&\frac{\mathbb{P}(Y = (\tilde{s},(y_{s'})_{s' \in \tilde{s}})|S = s)}{\mathbb{P}(Y = (\tilde{s},(y_{s'})_{s' \in \tilde{s}})|S = s')}\\
&= \left\{\begin{array}{ll}
1, & \textrm{if $s = s'$ or $s,s' \notin \tilde{s}$,}\\
\frac{\kappa(1-\lambda)T_s}{\lambda(1-\kappa)\hat{T}_s}, & \textrm{if $s \in \tilde{s} \not \ni s'$,}\\
\frac{\lambda(1-\kappa)\hat{T}_{s'}}{\kappa(1-\lambda)T_{s'}}, & \textrm{if $s \notin \tilde{s} \ni s'$,}\\
\frac{T_s\hat{T}_{s'}}{\hat{T}_{s}T_{s'}}, & \textrm{if $s,s' \in \tilde{s}$ and $s \neq s'$.}
\end{array}\right.
\end{align}
By Lemma \ref{lem:ldp}, one has $\frac{T_s}{\hat{T}_{s}},\frac{\hat{T}_s}{T_{s}} \leq 1+\frac{\textrm{e}^{\delta_s}-1}{2}||P_{\mathcal{U}|s}-\hat{P}_{\mathcal{U}_s}||_1 \leq \textrm{e}^{\varepsilon_2}$ for each $s \in \mathcal{S}$. Since $\frac{\lambda(1-\kappa)}{\kappa(1-\lambda)} \leq \frac{\kappa(1-\lambda)}{\lambda(1-\kappa)} \leq \textrm{e}^{\varepsilon_1}$, it follows that 
\begin{equation}
\frac{\mathbb{P}(Y = (\tilde{s},(y_{s'})_{s' \in \tilde{s}})|S = s)}{\mathbb{P}(Y = (\tilde{s},(y_{s'})_{s' \in \tilde{s}})|S = s')} \leq \max\left\{\textrm{e}^{\varepsilon_1+\varepsilon_2},\textrm{e}^{2\varepsilon_2}\right\},
\end{equation}
which proves the SLDP.
\end{proof}

This theorem shows that, similar to GRR-CR, the privacy requirements on the $\mathcal{R}^s$ become less strict as $P_{\mathcal{U}|s}$ and $\hat{P}_{\mathcal{U}|s}$ are closer. As for utility, we find the following theorem:

\begin{theorem} \label{thm:ueuti}
For any $P$ one has 
\begin{equation}
\recht{I}_P(\recht{UE-CR}(X);X) = \recht{I}_P(\recht{UE}(S);S) + \kappa\recht{I}_P(\mathcal{R}^S(U);U|S).
\end{equation}
\end{theorem}

\begin{proof}
One has
\begin{align}
\recht{I}_P(\recht{UE-CR}(X);X) &= \recht{I}_P(\tilde{S},(Y_{s'})_{s' \in \tilde{S}};S,U) \\
&= \recht{I}_P(\tilde{S};S,U) + \recht{I}_P((Y_{s'})_{s' \in \tilde{S}};S|\tilde{S})+\recht{I}_P((Y_{s'})_{s' \in \tilde{S}};U|S,\tilde{S}).
\end{align}
Similar to the proof of Theorem \ref{thm:grrcruti} we have that given $\tilde{S}$, the random variables $(Y_{s'})_{s' \in \tilde{S}}$ and $S$ are independent, hence $\recht{I}_P((Y_{s'})_{s' \in \tilde{S}};S|\tilde{S}) = 0$. Furthermore, $\tilde{S}$ and $U$ are independent given $S$, hence $\recht{I}_P(\tilde{S};S,U) = \recht{I}_P(\tilde{S};S)$. Furthermore, we can write
\begin{equation}
\recht{I}_P((Y_{s'})_{s' \in \tilde{S}};U|S,\tilde{S}) = \mathbb{E}_{s,\tilde{s}}\left[\recht{I}_P(Y_{s'})_{s' \in \tilde{S}};U|S=s,\tilde{S}=\tilde{s})\right].
\end{equation}
If $s \notin \tilde{s}$, then $U$ and $(Y_{s'})_{s' \in \tilde{S}}$ are independent given $S=s$ and $\tilde{S}=\tilde{s}$. If $s \in \tilde{s}$, then $U$ and $Y_{s'}$ are independent given $S=s$ and $\tilde{S}=\tilde{s}$, for $s' \neq s$. It follows that
\begin{align}
\recht{I}_P(Y_{s'})_{s' \in \tilde{S}};U|S=s,\tilde{S}=\tilde{s}) = \left\{\begin{array}{ll}
\recht{I}_P(Y_s;U|S=s),& \textrm{ if $s \in \tilde{s}$},\\
0,& \textrm{ otherwise}.
\end{array}\right.
\end{align}
From this we conclude that
\begin{align}
\recht{I}_P((Y_{s'})_{s' \in \tilde{S}};U|S,\tilde{S}) &= \sum_s \mathbb{P}(s \in \tilde{S}|S = s)\mathbb{P}(S = s)\recht{I}_P(Y_s;U|S=s) \\
&= \kappa \recht{I}(Y_S;U|S). \qedhere
\end{align}
\end{proof}

As before, we can conclude from this that we should let all $\mathcal{R}^s$ be GRR. This leaves us to finding $\varepsilon_2$, $\kappa$ and $\lambda$, which is a three-dimensional optimisation problem.

\section{Experiments} \label{sec:exp}

In order to test the feasibility of the different methods we perform several experiments, both on synthetic and real data. Throughout, we take $\alpha = 0.05$ unless stated otherwise.

Throughout the experiments, we use $\recht{I}_{\hat{P}}(X;Y)$ as a utility metric, occasionally normalised by dividing by $\recht{H}(X)$. We use this rather than $\recht{I}_{P^{*}}(X;Y)$, as the aggregator only has access to the former. In fact, while $P^{*}$ is known for the synthetic data, this is not the case for real data, so we cannot even use $\recht{I}_{P^{*}}(X;Y)$ as a utility metric.

The outline of the remainder of this section is as follows. In Section~\ref{ssec:polyopt} we focus on the PolyOpt method. In Section~\ref{ssec:synthetic} we consider the impact of $a_1$ and $a_2$. In Section~\ref{ssec:parameter} we analyse the optimal value of the parameters of our protocols. In Section~\ref{ssec:diff} we investigate the difference between IR and GRR-CR. In Section~\ref{ssec:n} we analyze the role of $n$ and $\alpha$. In Section \ref{ssec:robust}, we investigate the difference $\recht{I}_{\hat{P}}(X;Y)-\recht{I}_{P^*(X;Y)}$ for synthetic data, to evaluate the robustness of the utility metric. Finally, in Section~\ref{ssec:adult} we consider real data. 

\subsection{PolyOpt} \label{ssec:polyopt}

We first perform experiments to test the utility of the PolyOpt method introduced in Section \ref{sec:poly}. We perform numerical experiments on synthetic data. For $a_1 = a_2 = 3$, we draw 200 distributions from the Jeffreys prior on the space of probability distributions on $\mathcal{X}$. For each distribution, we draw $n = 1000$ items from this distribution, and we demand robustness w.r.t. this observed distribution. For each observed distribution, for $\varepsilon \in [0.2,8]$, and for each protocol of PolyOpt, IR, GRR-CR, UE-CR and SRR\footnote{for SRR we ignore the value of $\alpha$.}, we calculate the normalised utility $\frac{\recht{I}_{\hat{P}}(X;Y)}{\recht{H}(X)}$, which we average over all distributions. As a reference we perform the same analysis on GRR, the LDP protocol that maximises mutual information for large $\varepsilon$. Since GRR satisfies $\varepsilon$-LDP, it certainly satisfies $(\varepsilon,\mathcal{F})$-RLDP for any $\mathcal{F}$.

\begin{figure}
    \centering
%     \begin{subfigure}[b]{\textwidth}
%     \centering
% 		\begin{tikzpicture} 
% 		\tikzstyle{every node}=[font=\small]
%     \begin{axis}[%
%     width = 10cm,
%     height = 1cm,
%     hide axis,
%     xmin=0,
%     xmax=0.1,
%     ymin=0,
%     ymax=0.1,
%     legend columns=-1
%     ]
%     \addlegendimage{red}
%     \addlegendentry{SR};
%     \addlegendimage{blue}
%     \addlegendentry{GRR-CR};
%     \addlegendimage{black}
%     \addlegendentry{UE-CR};
%     \addlegendimage{brown}
%     \addlegendentry{SRR};
%     \addlegendimage{gray}
%     \addlegendentry{GRR};
%     \addlegendimage{purple}
%     \addlegendentry{PolyOpt};
%     \end{axis}
% \end{tikzpicture}
% 	\end{subfigure}\\
    \begin{subfigure}[b]{0.48\textwidth} 
    \centering
		\begin{tikzpicture}
		\tikzstyle{every node}=[font=\small]
		\begin{axis}[
		width = 7cm,
		height = 4cm,
		axis lines = left,
		xmin = 0,
		xmax = 8,
		ymin = 0,
		ymax = 1,
		ylabel near ticks,
		xlabel = $\varepsilon$, 
		ylabel = $\frac{\recht{I}(X;Y)}{\recht{H}(X)}$]
		\addplot[color=red,mark=none] table[col sep = comma] {CSV/PolyOpt/SR.csv}; \label{pl:sr}
		\addplot[color=blue,mark=none] table[col sep = comma] {CSV/PolyOpt/GRRCR.csv}; \label{pl:grrcr}
		\addplot[color=black,mark=none] table[col sep = comma] {CSV/PolyOpt/UECR.csv}; \label{pl:uecr}
		\addplot[color=brown,mark=none] table[col sep = comma] {CSV/PolyOpt/SRR.csv}; \label{pl:srr}
		\addplot[color=gray,mark=none] table[col sep = comma] {CSV/PolyOpt/GRR.csv}; \label{pl:grr}
		\addplot[color=purple,mark=none] table[col sep = comma] {CSV/PolyOpt/PolyOpt.csv}; \label{pl:PolyOpt}
		\end{axis}
		\end{tikzpicture}
	\end{subfigure}
	\begin{subfigure}[b]{0.48\textwidth}
	\centering
		\begin{tikzpicture}
		\tikzstyle{every node}=[font=\small]
		\begin{axis}[
		width = 7cm,
		height = 4cm,
		axis lines = left,
		xmin = 0,
		xmax = 8,
		ymin = 0,
		ymax = 2.4,
		ylabel near ticks,
		xlabel = $\varepsilon$, 
		ylabel = time (s)]
		\addplot[color=red,mark=none] table[col sep = comma] {CSV/Time/SR.csv};
		\addplot[color=blue,mark=none] table[col sep = comma] {CSV/Time/GRRCR.csv};
		\addplot[color=black,mark=none] table[col sep = comma] {CSV/Time/UECR.csv};
		\addplot[color=purple,mark=none] table[col sep = comma] {CSV/Time/PolyOpt.csv};
		\end{axis}
		\end{tikzpicture}
		
	\end{subfigure}
    \caption{\it Experiments on synthetic data for $a_1 = a_2 = 3$ (\ref{pl:PolyOpt} PolyOpt, \ref{pl:sr} IR, \ref{pl:grrcr} GRR-CR, \ref{pl:uecr}~\mbox{UE-CR}, \ref{pl:srr} SRR, \ref{pl:grr} GRR).  \label{fig:opt}}
    
\end{figure}
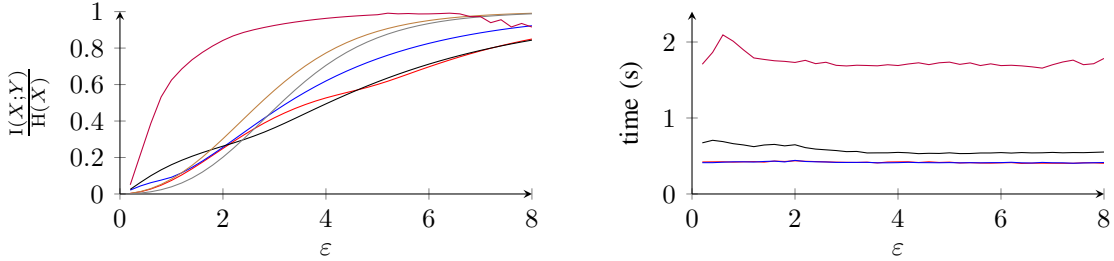

The results are in Figure \ref{fig:opt}. As we can see, PolyOpt significantly outperforms the other methods, although the optimisation we used (we used Matlab, specifically the MPT3 toolbox) becomes more inaccurate at larger $\varepsilon$. However, the downside of Polyopt lies in its computation time, which is significantly higher than that of other methods. All experiments were conducted on a PC with Intel Core i7-7700HQ 2.8GHz and 32GB memory. As can be observed in Figure~\ref{fig:opt} for larger $a$ the computation time increases dramatically: for $a_1 = 3$, $a_2 = 4$ the computation time is 72s on average, and for $a_1 = a_2 = 4$ we terminated the computation when it was still running after 12 hours.

In general, if the user has enough computation power to use the PolyOpt method, then this is recommended, because it clearly outperforms all other protocols. However, it is possible that this is computationally unfeasible. For most of our other experiments, we assume that this is the case, and we study the utility of the other methods.

\subsection{Synthetic data} \label{ssec:synthetic}

We perform the same procedure as before, but for different $a_1,a_2$. The results are in Figures \ref{fig:synthexp1} and \ref{fig:synthexp2}. As can be seen, for $\varepsilon$ large enough, SRR is the best protocol, which is remarkable as it has the strictest privacy requirement. The larger $a_1$ and $a_2$ are, the larger $\varepsilon$ has to be for SRR to become the preferred method. We see that IR and GRR-CR perform more or less similar. For small $a_1$ and $a_2$, we see that UE-CR outperforms these; for high $\varepsilon$, on the other hand, UE-CR is the worse choice. This is understandable considering the fact that UE yields less mutual information between input and output than GRR \cite{lopuhaa2020privacy}. 

\begin{figure}
    \centering
%     \begin{subfigure}[b]{\textwidth}
%     \centering
% 		\begin{tikzpicture} 
% 		\tikzstyle{every node}=[font=\small]
%     \begin{axis}[%
%     width = 1mm,
%     height = 1mm,
%     hide axis,
%     xmin=0,
%     xmax=0.1,
%     ymin=0,
%     ymax=0.1,
%     legend columns=-1
%     ]
%     \addlegendimage{red}
%     \addlegendentry{SR};
%     \addlegendimage{blue}
%     \addlegendentry{GRR-CR};
%     \addlegendimage{black}
%     \addlegendentry{UE-CR};
%     \addlegendimage{brown}
%     \addlegendentry{SRR};
%     \addlegendimage{gray}
%     \addlegendentry{GRR};
%     \end{axis}
% \end{tikzpicture}
% 	\end{subfigure}\\
    \begin{subfigure}[b]{0.45\textwidth}
    \centering
		\begin{tikzpicture}
		\tikzstyle{every node}=[font=\small]
		\begin{axis}[
		width = 7cm,
		height = 4cm,
		axis lines = left,
		xmin = 0,
		xmax = 8,
		ymin = 0,
		ymax = 1,
		ylabel near ticks,
		xlabel = {$\varepsilon \  (a_1 = 3, a_2 = 3)$}, 
		ylabel = $\frac{\recht{I}(X;Y)}{\recht{H}(X)}$]
		\addplot[color=red,mark=none] table[col sep = comma] {CSV/S3U3/SR.csv};
		\addplot[color=blue,mark=none] table[col sep = comma] {CSV/S3U3/GRRCR.csv};
		\addplot[color=black,mark=none] table[col sep = comma] {CSV/S3U3/UECR.csv};
		\addplot[color=brown,mark=none] table[col sep = comma] {CSV/S3U3/SRR.csv};
		\addplot[color=gray,mark=none] table[col sep = comma] {CSV/S3U3/GRR.csv};
		\end{axis}
		\end{tikzpicture}
	\end{subfigure}
	\begin{subfigure}[b]{0.45\textwidth}
    \centering
		\begin{tikzpicture}
		\tikzstyle{every node}=[font=\small]
		\begin{axis}[
		width = 7cm,
		height = 4cm,
		axis lines = left,
		xmin = 0,
		xmax = 8,
		ymin = 0,
		ymax = 1,
		ylabel near ticks,
		xlabel = {$\varepsilon \  (a_1 = 3, a_2 = 5)$}, 
		ylabel = $\frac{\recht{I}(X;Y)}{\recht{H}(X)}$]
		\addplot[color=red,mark=none] table[col sep = comma] {CSV/S3U5/SR.csv};
		\addplot[color=blue,mark=none] table[col sep = comma] {CSV/S3U5/GRRCR.csv};
		\addplot[color=black,mark=none] table[col sep = comma] {CSV/S3U5/UECR.csv};
		\addplot[color=brown,mark=none] table[col sep = comma] {CSV/S3U5/SRR.csv};
		\addplot[color=gray,mark=none] table[col sep = comma] {CSV/S3U5/GRR.csv};
		\end{axis}
		\end{tikzpicture}
	\end{subfigure}
    \begin{subfigure}[b]{0.45\textwidth}
    \centering
		\begin{tikzpicture}
		\tikzstyle{every node}=[font=\small]
		\begin{axis}[
		width = 7cm,
		height = 4cm,
		axis lines = left,
		xmin = 0,
		xmax = 8,
		ymin = 0,
		ymax = 1,
		ylabel near ticks,
		xlabel = {$\varepsilon \  (a_1 = 3, a_2 = 7)$}, 
		ylabel = $\frac{\recht{I}(X;Y)}{\recht{H}(X)}$]
		\addplot[color=red,mark=none] table[col sep = comma] {CSV/S3U7/SR.csv};
		\addplot[color=blue,mark=none] table[col sep = comma] {CSV/S3U7/GRRCR.csv};
		\addplot[color=black,mark=none] table[col sep = comma] {CSV/S3U7/UECR.csv};
		\addplot[color=brown,mark=none] table[col sep = comma] {CSV/S3U7/SRR.csv};
		\addplot[color=gray,mark=none] table[col sep = comma] {CSV/S3U7/GRR.csv};
		\end{axis}
		\end{tikzpicture}
	\end{subfigure}
	\begin{subfigure}[b]{0.45\textwidth}
    \centering
		\begin{tikzpicture}
		\tikzstyle{every node}=[font=\small]
		\begin{axis}[
		width = 7cm,
		height = 4cm,
		axis lines = left,
		xmin = 0,
		xmax = 8,
		ymin = 0,
		ymax = 1,
		ylabel near ticks,
		xlabel = {$\varepsilon \  (a_1 = 3, a_2 = 9)$}, 
		ylabel = $\frac{\recht{I}(X;Y)}{\recht{H}(X)}$]
		\addplot[color=red,mark=none] table[col sep = comma] {CSV/S3U9/SR.csv};
		\addplot[color=blue,mark=none] table[col sep = comma] {CSV/S3U9/GRRCR.csv};
		\addplot[color=black,mark=none] table[col sep = comma] {CSV/S3U9/UECR.csv};
		\addplot[color=brown,mark=none] table[col sep = comma] {CSV/S3U9/SRR.csv};
		\addplot[color=gray,mark=none] table[col sep = comma] {CSV/S3U9/GRR.csv};
		\end{axis}
		\end{tikzpicture}
	\end{subfigure}
	    \begin{subfigure}[b]{0.45\textwidth}
    \centering
		\begin{tikzpicture}
		\tikzstyle{every node}=[font=\small]
		\begin{axis}[
		width = 7cm,
		height = 4cm,
		axis lines = left,
		xmin = 0,
		xmax = 8,
		ymin = 0,
		ymax = 1,
		ylabel near ticks,
		xlabel = {$\varepsilon \  (a_1 = 5, a_2 = 3)$}, 
		ylabel = $\frac{\recht{I}(X;Y)}{\recht{H}(X)}$]
		\addplot[color=red,mark=none] table[col sep = comma] {CSV/S5U3/SR.csv};
		\addplot[color=blue,mark=none] table[col sep = comma] {CSV/S5U3/GRRCR.csv};
		\addplot[color=black,mark=none] table[col sep = comma] {CSV/S5U3/UECR.csv};
		\addplot[color=brown,mark=none] table[col sep = comma] {CSV/S5U3/SRR.csv};
		\addplot[color=gray,mark=none] table[col sep = comma] {CSV/S5U3/GRR.csv};
		\end{axis}
		\end{tikzpicture}
	\end{subfigure}
	\begin{subfigure}[b]{0.45\textwidth}
    \centering
		\begin{tikzpicture}
		\tikzstyle{every node}=[font=\small]
		\begin{axis}[
		width = 7cm,
		height = 4cm,
		axis lines = left,
		xmin = 0,
		xmax = 8,
		ymin = 0,
		ymax = 1,
		ylabel near ticks,
		xlabel = {$\varepsilon \  (a_1 = 5, a_2 = 5)$}, 
		ylabel = $\frac{\recht{I}(X;Y)}{\recht{H}(X)}$]
		\addplot[color=red,mark=none] table[col sep = comma] {CSV/S5U5/SR.csv};
		\addplot[color=blue,mark=none] table[col sep = comma] {CSV/S5U5/GRRCR.csv};
		\addplot[color=black,mark=none] table[col sep = comma] {CSV/S5U5/UECR.csv};
		\addplot[color=brown,mark=none] table[col sep = comma] {CSV/S5U5/SRR.csv};
		\addplot[color=gray,mark=none] table[col sep = comma] {CSV/S5U5/GRR.csv};
		\end{axis}
		\end{tikzpicture}
	\end{subfigure}
    \begin{subfigure}[b]{0.45\textwidth}
    \centering
		\begin{tikzpicture}
		\tikzstyle{every node}=[font=\small]
		\begin{axis}[
		width = 7cm,
		height = 4cm,
		axis lines = left,
		xmin = 0,
		xmax = 8,
		ymin = 0,
		ymax = 1,
		ylabel near ticks,
		xlabel = {$\varepsilon \  (a_1 = 5, a_2 = 7)$}, 
		ylabel = $\frac{\recht{I}(X;Y)}{\recht{H}(X)}$]
		\addplot[color=red,mark=none] table[col sep = comma] {CSV/S5U7/SR.csv};
		\addplot[color=blue,mark=none] table[col sep = comma] {CSV/S5U7/GRRCR.csv};
		\addplot[color=black,mark=none] table[col sep = comma] {CSV/S5U7/UECR.csv};
		\addplot[color=brown,mark=none] table[col sep = comma] {CSV/S5U7/SRR.csv};
		\addplot[color=gray,mark=none] table[col sep = comma] {CSV/S5U7/GRR.csv};
		\end{axis}
		\end{tikzpicture}
	\end{subfigure}
	\begin{subfigure}[b]{0.45\textwidth}
    \centering
		\begin{tikzpicture}
		\tikzstyle{every node}=[font=\small]
		\begin{axis}[
		width = 7cm,
		height = 4cm,
		axis lines = left,
		xmin = 0,
		xmax = 8,
		ymin = 0,
		ymax = 1,
		ylabel near ticks,
		xlabel = {$\varepsilon \  (a_1 = 5, a_2 = 9)$}, 
		ylabel = $\frac{\recht{I}(X;Y)}{\recht{H}(X)}$]
		\addplot[color=red,mark=none] table[col sep = comma] {CSV/S5U9/SR.csv};
		\addplot[color=blue,mark=none] table[col sep = comma] {CSV/S5U9/GRRCR.csv};
		\addplot[color=black,mark=none] table[col sep = comma] {CSV/S5U9/UECR.csv};
		\addplot[color=brown,mark=none] table[col sep = comma] {CSV/S5U9/SRR.csv};
		\addplot[color=gray,mark=none] table[col sep = comma] {CSV/S5U9/GRR.csv};
		\end{axis}
		\end{tikzpicture}
	\end{subfigure}
    \caption{\it Experiments on synthetic data, $a_1 \in\{3,5\}$, $a_2 \in \{3,5,7,9\}$ (\ref{pl:sr} IR, \ref{pl:grrcr} GRR-CR, \ref{pl:uecr}~\mbox{UE-CR}, \ref{pl:srr} SRR, \ref{pl:grr} GRR).\label{fig:synthexp1}}
\end{figure}
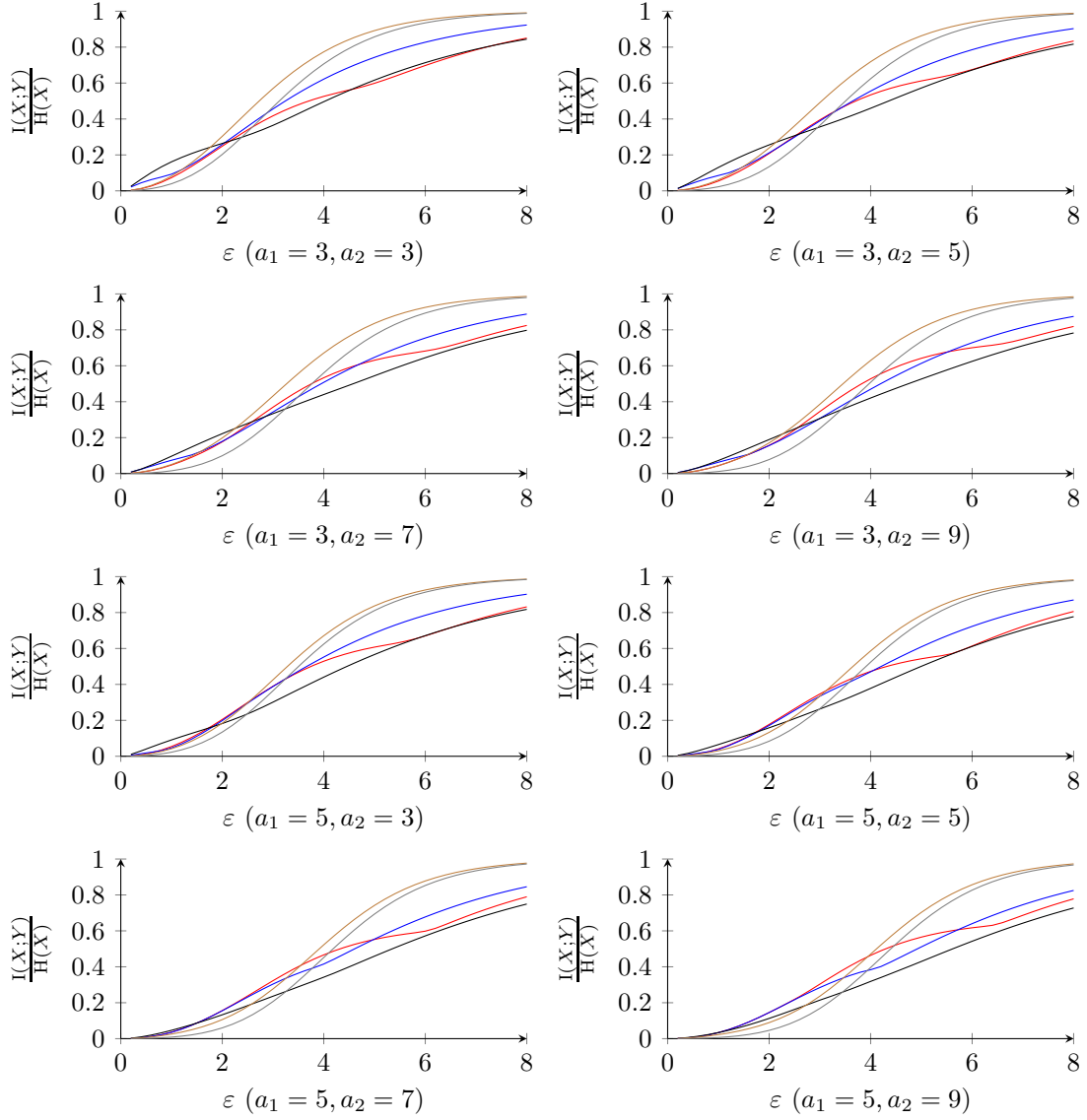

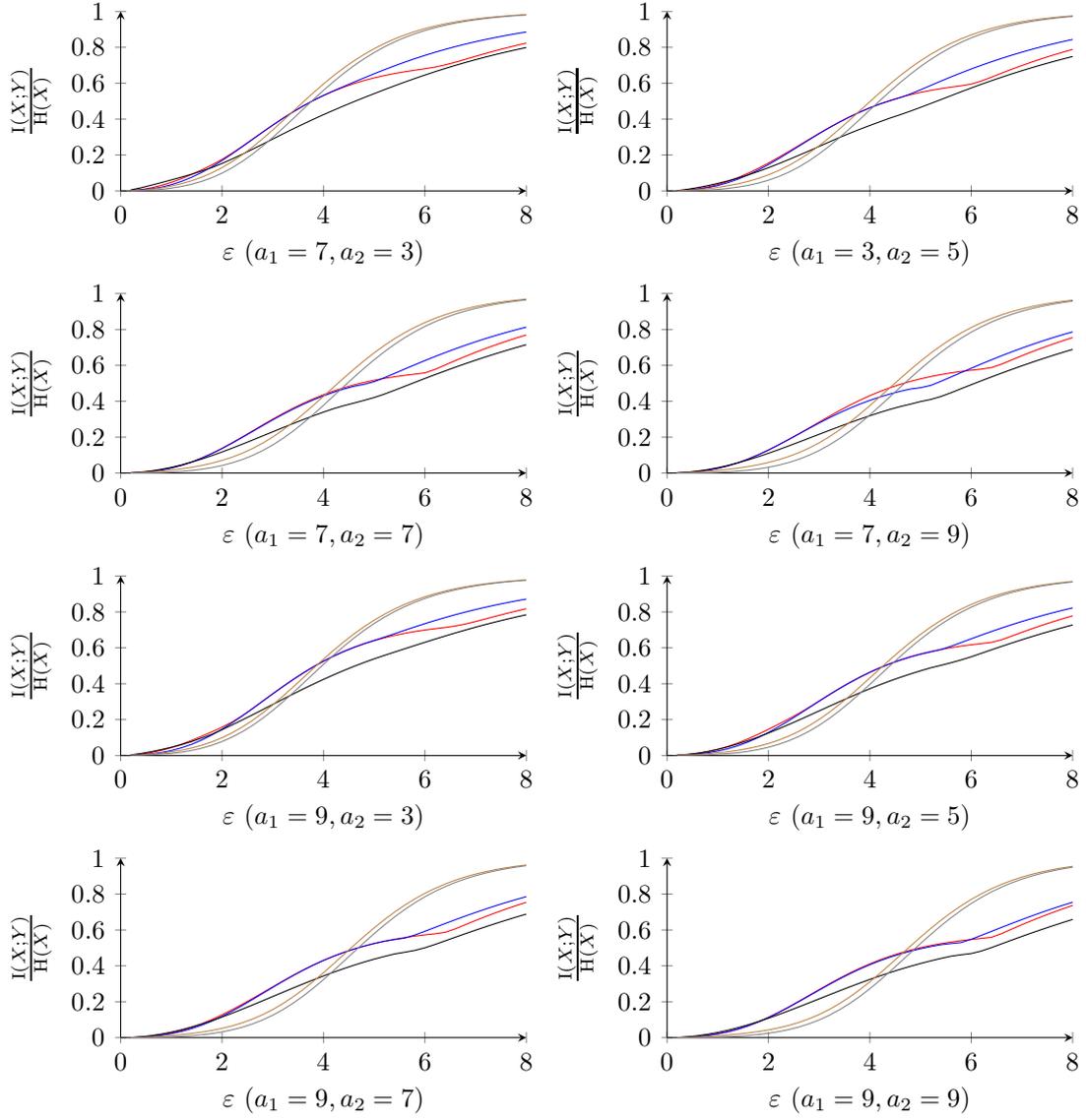
\begin{figure}
    \centering
%     \begin{subfigure}[b]{\textwidth}
%     \centering
% 		\begin{tikzpicture} 
% 		\tikzstyle{every node}=[font=\small]
%     \begin{axis}[%
%     width = 1mm,
%     height = 1mm,
%     hide axis,
%     xmin=0,
%     xmax=0.1,
%     ymin=0,
%     ymax=0.1,
%     legend columns=-1
%     ]
%     \addlegendimage{red}
%     \addlegendentry{SR};
%     \addlegendimage{blue}
%     \addlegendentry{GRR-CR};
%     \addlegendimage{black}
%     \addlegendentry{UE-CR};
%     \addlegendimage{brown}
%     \addlegendentry{SRR};
%     \addlegendimage{gray}
%     \addlegendentry{GRR};
%     \end{axis}
% \end{tikzpicture}
% 	\end{subfigure}\\
    \begin{subfigure}[b]{0.45\textwidth}
    \centering
		\begin{tikzpicture}
		\tikzstyle{every node}=[font=\small]
		\begin{axis}[
		width = 7cm,
		height = 4cm,
		axis lines = left,
		xmin = 0,
		xmax = 8,
		ymin = 0,
		ymax = 1,
		ylabel near ticks,
		xlabel = {$\varepsilon \  (a_1 = 7, a_2 = 3)$}, 
		ylabel = $\frac{\recht{I}(X;Y)}{\recht{H}(X)}$]
		\addplot[color=red,mark=none] table[col sep = comma] {CSV/S7U3/SR.csv};
		\addplot[color=blue,mark=none] table[col sep = comma] {CSV/S7U3/GRRCR.csv};
		\addplot[color=black,mark=none] table[col sep = comma] {CSV/S7U3/UECR.csv};
		\addplot[color=brown,mark=none] table[col sep = comma] {CSV/S7U3/SRR.csv};
		\addplot[color=gray,mark=none] table[col sep = comma] {CSV/S7U3/GRR.csv};
		\end{axis}
		\end{tikzpicture}
	\end{subfigure}
	\begin{subfigure}[b]{0.45\textwidth}
    \centering
		\begin{tikzpicture}
		\tikzstyle{every node}=[font=\small]
		\begin{axis}[
		width = 7cm,
		height = 4cm,
		axis lines = left,
		xmin = 0,
		xmax = 8,
		ymin = 0,
		ymax = 1,
		ylabel near ticks,
		xlabel = {$\varepsilon \  (a_1 = 3, a_2 = 5)$}, 
		ylabel = $\frac{\recht{I}(X;Y)}{\recht{H}(X)}$]
		\addplot[color=red,mark=none] table[col sep = comma] {CSV/S7U5/SR.csv};
		\addplot[color=blue,mark=none] table[col sep = comma] {CSV/S7U5/GRRCR.csv};
		\addplot[color=black,mark=none] table[col sep = comma] {CSV/S7U5/UECR.csv};
		\addplot[color=brown,mark=none] table[col sep = comma] {CSV/S7U5/SRR.csv};
		\addplot[color=gray,mark=none] table[col sep = comma] {CSV/S7U5/GRR.csv};
		\end{axis}
		\end{tikzpicture}
	\end{subfigure}
    \begin{subfigure}[b]{0.45\textwidth}
    \centering
		\begin{tikzpicture}
		\tikzstyle{every node}=[font=\small]
		\begin{axis}[
		width = 7cm,
		height = 4cm,
		axis lines = left,
		xmin = 0,
		xmax = 8,
		ymin = 0,
		ymax = 1,
		ylabel near ticks,
		xlabel = {$\varepsilon \  (a_1 = 7, a_2 = 7)$}, 
		ylabel = $\frac{\recht{I}(X;Y)}{\recht{H}(X)}$]
		\addplot[color=red,mark=none] table[col sep = comma] {CSV/S7U7/SR.csv};
		\addplot[color=blue,mark=none] table[col sep = comma] {CSV/S7U7/GRRCR.csv};
		\addplot[color=black,mark=none] table[col sep = comma] {CSV/S7U7/UECR.csv};
		\addplot[color=brown,mark=none] table[col sep = comma] {CSV/S7U7/SRR.csv};
		\addplot[color=gray,mark=none] table[col sep = comma] {CSV/S7U7/GRR.csv};
		\end{axis}
		\end{tikzpicture}
	\end{subfigure}
	\begin{subfigure}[b]{0.45\textwidth}
    \centering
		\begin{tikzpicture}
		\tikzstyle{every node}=[font=\small]
		\begin{axis}[
		width = 7cm,
		height = 4cm,
		axis lines = left,
		xmin = 0,
		xmax = 8,
		ymin = 0,
		ymax = 1,
		ylabel near ticks,
		xlabel = {$\varepsilon \  (a_1 = 7, a_2 = 9)$}, 
		ylabel = $\frac{\recht{I}(X;Y)}{\recht{H}(X)}$]
		\addplot[color=red,mark=none] table[col sep = comma] {CSV/S7U9/SR.csv};
		\addplot[color=blue,mark=none] table[col sep = comma] {CSV/S7U9/GRRCR.csv};
		\addplot[color=black,mark=none] table[col sep = comma] {CSV/S7U9/UECR.csv};
		\addplot[color=brown,mark=none] table[col sep = comma] {CSV/S7U9/SRR.csv};
		\addplot[color=gray,mark=none] table[col sep = comma] {CSV/S7U9/GRR.csv};
		\end{axis}
		\end{tikzpicture}
	\end{subfigure}
	    \begin{subfigure}[b]{0.45\textwidth}
    \centering
		\begin{tikzpicture}
		\tikzstyle{every node}=[font=\small]
		\begin{axis}[
		width = 7cm,
		height = 4cm,
		axis lines = left,
		xmin = 0,
		xmax = 8,
		ymin = 0,
		ymax = 1,
		ylabel near ticks,
		xlabel = {$\varepsilon \  (a_1 = 9, a_2 = 3)$}, 
		ylabel = $\frac{\recht{I}(X;Y)}{\recht{H}(X)}$]
		\addplot[color=red,mark=none] table[col sep = comma] {CSV/S9U3/SR.csv};
		\addplot[color=blue,mark=none] table[col sep = comma] {CSV/S9U3/GRRCR.csv};
		\addplot[color=black,mark=none] table[col sep = comma] {CSV/S9U3/UECR.csv};
		\addplot[color=brown,mark=none] table[col sep = comma] {CSV/S9U3/SRR.csv};
		\addplot[color=gray,mark=none] table[col sep = comma] {CSV/S9U3/GRR.csv};
		\end{axis}
		\end{tikzpicture}
	\end{subfigure}
	\begin{subfigure}[b]{0.45\textwidth}
    \centering
		\begin{tikzpicture}
		\tikzstyle{every node}=[font=\small]
		\begin{axis}[
		width = 7cm,
		height = 4cm,
		axis lines = left,
		xmin = 0,
		xmax = 8,
		ymin = 0,
		ymax = 1,
		ylabel near ticks,
		xlabel = {$\varepsilon \  (a_1 = 9, a_2 = 5)$}, 
		ylabel = $\frac{\recht{I}(X;Y)}{\recht{H}(X)}$]
		\addplot[color=red,mark=none] table[col sep = comma] {CSV/S9U5/SR.csv};
		\addplot[color=blue,mark=none] table[col sep = comma] {CSV/S9U5/GRRCR.csv};
		\addplot[color=black,mark=none] table[col sep = comma] {CSV/S9U5/UECR.csv};
		\addplot[color=brown,mark=none] table[col sep = comma] {CSV/S9U5/SRR.csv};
		\addplot[color=gray,mark=none] table[col sep = comma] {CSV/S9U5/GRR.csv};
		\end{axis}
		\end{tikzpicture}
	\end{subfigure}
    \begin{subfigure}[b]{0.45\textwidth}
    \centering
		\begin{tikzpicture}
		\tikzstyle{every node}=[font=\small]
		\begin{axis}[
		width = 7cm,
		height = 4cm,
		axis lines = left,
		xmin = 0,
		xmax = 8,
		ymin = 0,
		ymax = 1,
		ylabel near ticks,
		xlabel = {$\varepsilon \  (a_1 = 9, a_2 = 7)$}, 
		ylabel = $\frac{\recht{I}(X;Y)}{\recht{H}(X)}$]
		\addplot[color=red,mark=none] table[col sep = comma] {CSV/S9U7/SR.csv};
		\addplot[color=blue,mark=none] table[col sep = comma] {CSV/S9U7/GRRCR.csv};
		\addplot[color=black,mark=none] table[col sep = comma] {CSV/S9U7/UECR.csv};
		\addplot[color=brown,mark=none] table[col sep = comma] {CSV/S9U7/SRR.csv};
		\addplot[color=gray,mark=none] table[col sep = comma] {CSV/S9U7/GRR.csv};
		\end{axis}
		\end{tikzpicture}
	\end{subfigure}
	\begin{subfigure}[b]{0.45\textwidth}
    \centering
		\begin{tikzpicture}
		\tikzstyle{every node}=[font=\small]
		\begin{axis}[
		width = 7cm,
		height = 4cm,
		axis lines = left,
		xmin = 0,
		xmax = 8,
		ymin = 0,
		ymax = 1,
		ylabel near ticks,
		xlabel = {$\varepsilon \  (a_1 = 9, a_2 = 9)$}, 
		ylabel = $\frac{\recht{I}(X;Y)}{\recht{H}(X)}$]
		\addplot[color=red,mark=none] table[col sep = comma] {CSV/S9U9/SR.csv};
		\addplot[color=blue,mark=none] table[col sep = comma] {CSV/S9U9/GRRCR.csv};
		\addplot[color=black,mark=none] table[col sep = comma] {CSV/S9U9/UECR.csv};
		\addplot[color=brown,mark=none] table[col sep = comma] {CSV/S9U9/SRR.csv};
		\addplot[color=gray,mark=none] table[col sep = comma] {CSV/S9U9/GRR.csv};
		\end{axis}
		\end{tikzpicture}
	\end{subfigure}
    \caption{\it Experiments on synthetic data, $a_1 \in\{7,9\}$, $a_2 \in \{3,5,7,9\}$ (\ref{pl:sr} IR, \ref{pl:grrcr} GRR-CR, \ref{pl:uecr}~\mbox{UE-CR}, \ref{pl:srr} SRR, \ref{pl:grr} GRR).\label{fig:synthexp2}}
\end{figure}

Looking at GRR, we see that it performs slightly worse than SRR across all $\varepsilon$. This is expected behaviour since the protocols are very similar, but SRR is better tailored to the Privacy Funnel scenario.

\subsection{Optimal parameter settings} \label{ssec:parameter}

We plot the values of $\varepsilon_2/\varepsilon$ for IR and GRR-CR, and $\kappa$ and $\lambda$ for UE-CR, to get insight into the ideal parameter settings. We take $a_1 = a_2 = 5$, and we draw two distributions from the Jeffreys prior ($n = 1000$). We also draw 200 distributions, and take the average parameter settings over these distributions. These graphs are depicted in Figure \ref{fig:param}.

As we can see from the samples, for low $\varepsilon$ it is optimal to take either $\varepsilon_2 = 0$ or $\varepsilon_2 = \varepsilon$; in IR and GRR-CR this means to use all the privacy budget for either transmitting $S$ or $U$. Note that when the whole privacy budget is spent on $S$, then IR and GRR-CR are the same protocol; this explains why they behave so similar for low $\varepsilon$. Furthermore, we see that for GRR-CR it is beneficial for any $P$ to spend the entire privacy budget on $U$ for $\varepsilon < 1$, and on $S$ for slightly higher $1 < \varepsilon < 3$. By contrast, it depends on $P$ whether the privacy budget of IR is to be spent on $S$ or on $U$ for low $\varepsilon$. This explains why the average value of $\varepsilon_2/\varepsilon$ is close to 0.5 regardless of the value of $\varepsilon$ for IR.

For UE-CR we also see that the privacy budget is spent only on one of the two components: for low $\varepsilon$, it is optimal to take $\kappa = \lambda = 1$, which means there is no information leakage about $S$. It is only when $\varepsilon$ grows larger that it becomes optimal to divide the privacy budget among $S$ and $U$. The point where such a division is optimal depends on the distribution.

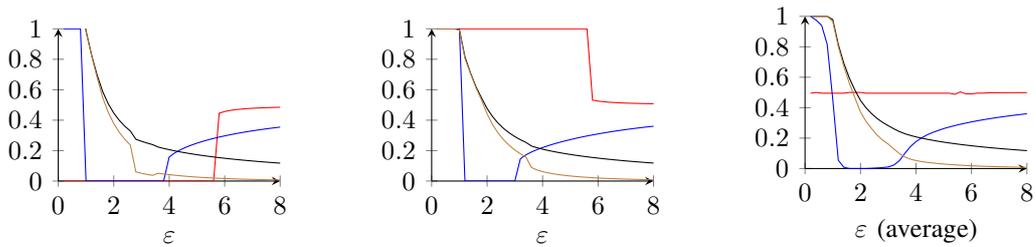
\begin{figure}
    \centering
%     \begin{subfigure}[b]{\textwidth}
%     \centering
% 		\begin{tikzpicture} 
% 		\tikzstyle{every node}=[font=\small]
%     \begin{axis}[%
%     width = 1mm,
%     height = 1mm,
%     hide axis,
%     xmin=0,
%     xmax=0.1,
%     ymin=0,
%     ymax=0.1,
%     legend columns=-1
%     ]
%     \addlegendimage{red}
%     \addlegendentry{$\varepsilon_2/\varepsilon$ (SR)};
%     \addlegendimage{blue}
%     \addlegendentry{$\varepsilon_2/\varepsilon$ (GRR-CR)};
%     \addlegendimage{black}
%     \addlegendentry{$\kappa$ (UE-CR)};
%     \addlegendimage{brown}
%     \addlegendentry{$\lambda$ (UE-CR)};
%     \end{axis}
% \end{tikzpicture}
% 	\end{subfigure}\\
    \begin{subfigure}[b]{0.3\textwidth}
		\begin{tikzpicture}
		\tikzstyle{every node}=[font=\small]
		\begin{axis}[
		width = 4.5cm,
		height = 3.6cm,
		axis lines = left,
		xmin = 0,
		xmax = 8,
		ymin = 0,
		ymax = 1,
		xlabel = $\varepsilon$]
		\addplot[color=red,mark=none] table[col sep = comma] {CSV/param1/SR.csv};
		\addplot[color=blue,mark=none] table[col sep = comma] {CSV/param1/GRRCR.csv};
		\addplot[color=black,mark=none] table[col sep = comma] {CSV/param1/kappa.csv};
		\addplot[color=brown,mark=none] table[col sep = comma] {CSV/param1/lambda.csv};
		\end{axis}
		\end{tikzpicture}
	\end{subfigure}
	\begin{subfigure}[b]{0.3\textwidth}
		\begin{tikzpicture}
		\tikzstyle{every node}=[font=\small]
		\begin{axis}[
		width = 4.5cm,
		height = 3.6cm,
		axis lines = left,
		xmin = 0,
		xmax = 8,
		ymin = 0,
		ymax = 1,
		xlabel = $\varepsilon$]
		\addplot[color=red,mark=none] table[col sep = comma] {CSV/param2/SR.csv};
		\addplot[color=blue,mark=none] table[col sep = comma] {CSV/param2/GRRCR.csv};
		\addplot[color=black,mark=none] table[col sep = comma] {CSV/param2/kappa.csv};
		\addplot[color=brown,mark=none] table[col sep = comma] {CSV/param2/lambda.csv};
		\end{axis}
		\end{tikzpicture}
	\end{subfigure}
	\begin{subfigure}[b]{0.3\textwidth}
		\begin{tikzpicture}
		\tikzstyle{every node}=[font=\small]
		\begin{axis}[
		width = 4.5cm,
		height = 3.6cm,
		axis lines = left,
		xmin = 0,
		xmax = 8,
		ymin = 0,
		ymax = 1,
		xlabel = {$\varepsilon$ (average)}]
		\addplot[color=red,mark=none] table[col sep = comma] {CSV/paramean/SR.csv};
		\addplot[color=blue,mark=none] table[col sep = comma] {CSV/paramean/GRRCR.csv};
		\addplot[color=black,mark=none] table[col sep = comma] {CSV/paramean/kappa.csv};
		\addplot[color=brown,mark=none] table[col sep = comma] {CSV/paramean/lambda.csv};
		\end{axis}
		\end{tikzpicture}
	\end{subfigure}
	\vspace{-1.5em}
    \caption{\it Parameter values for IR and CR ($a_1 = a_2 = 5$) for two distributions and the average over 200 distributions (\ref{pl:sr} $\varepsilon_2/\varepsilon$ (IR), \ref{pl:grrcr} $\varepsilon_2/\varepsilon$ (GRR-CR), \ref{pl:uecr} $\kappa$ (UE-CR), \ref{pl:srr} $\lambda$ (UE-CR)).}
    \label{fig:param}
\end{figure}

\subsection{IR vs GRR-CR} \label{ssec:diff}

We also try to find out what causes the difference between IR and GRR-CR. In Figure~\ref{fig:IRvsGRRCR} we plot the normalised difference in utility between these two protocols against $\recht{H}(X)$, for 2000 randomly generated probability distributions (with $a_1=a_2=5$ and $\varepsilon = 4$). As one can see, there are many distributions where the two have equal utility, which is caused by the fact that the two protocols coincide when the whole privacy budget is allotted to $S$. Among the other distributions, however, we see a downward trend signifying that GRR-CR outperforms IR for large $\recht{H}(X)$.

\begin{figure}
    \centering
		\begin{tikzpicture}
		\tikzstyle{every node}=[font=\small]
		\begin{axis}[
		width = 8cm,
		height = 4cm,
		xmin = 1.5,
		xmax = 3,
		axis lines = left,
		ylabel near ticks,
		xlabel = {$\recht{H}(X)$}, 
		ylabel = {norm. uti. diff.},]
		\addplot[only marks,mark=+] table[col sep = comma] {CSV/SRvsGRRCR/plot.csv};
		\end{axis}
		\end{tikzpicture}
    \caption{Difference in $\frac{\recht{I}(X;Y)}{\recht{H}(X)}$ between SR and GRR-CR, plotted against $\recht{H}(X)$, for 2000 distributions drawn from the Jeffreys prior for $a_1=a_2=5$; $\varepsilon = 4$. \label{fig:IRvsGRRCR}}
\end{figure}
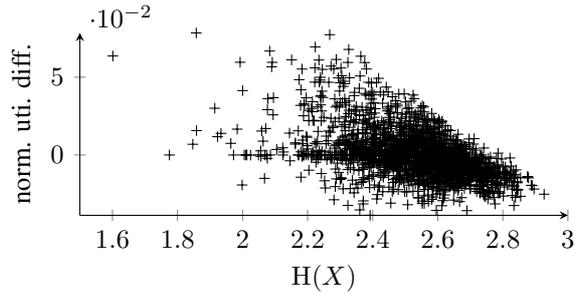

\subsection{Role of $n$ and $\alpha$} \label{ssec:n}

We also vary $n$ and $\alpha$, which were taken to be $1000$ and $0.05$ before, respectively. Taking larger $n$ and smaller $\alpha$ have the same effect, namely that $\mathcal{F}$ is smaller. As can be seen from Figure \ref{fig:changen}, taking larger $n$ has no effect on SRR and GRR, as they do not depend on $\mathcal{F}$. For IR, GRR-CR and UE-CR, we see that the larger $n$ is, the better utility they provide. This is more pronounced for GRR-CR and UE-CR than it is for SR, which can be explained from the fact that the privacy parameter $\delta_2$ from Theorem \ref{thm:irpriv} does not only depend on the size of $\mathcal{F}$, but also on $\max_{s,s'} ||\hat{P}_{\mathcal{U}|s}-\hat{P}_{\mathcal{U}|s'}||_1$. As such, the increase in utility that comes from reducing $\mathcal{F}$ is more limited than with CR.

We also look at the effect of $\alpha$ on the utility of IR, GRR-CR, and UE-CR. As mentioned before, the smaller $\alpha$, the larger $\mathcal{F}$, and the less utility the protocols will provide. This is reflected in Figure \ref{fig:changealpha}, where we see that having a smaller $\alpha$ reduces the utility of GRR-CR and UE-CR (we take $a_1 = a_2 = 5$ and $n = 1000$, and take the average over $200$ distributions). For IR there is no difference at all: this is because the maximum $d=2$ is obtained in Theorem \ref{thm:irpriv}, at which point the protocol is not affected by changing $\alpha$. For CR the loss of utility caused by changing $\alpha$ is rather small in absolute terms, but becomes important for low $\varepsilon$, as a change from $\alpha$ from $0.1$ to $0.0001$ causes an average utility loss of 36\% for GRR-CR, and 47\% for UE-CR for $\varepsilon = 0.1$.

\begin{figure}
    \centering
%     \begin{subfigure}[b]{\textwidth}
%     \centering
% 		\begin{tikzpicture} 
% 		\tikzstyle{every node}=[font=\small]
%     \begin{axis}[%
%     width = 1mm,
%     height = 1mm,
%     hide axis,
%     xmin=0,
%     xmax=0.1,
%     ymin=0,
%     ymax=0.1,
%     legend columns=-1
%     ]
%     \addlegendimage{red}
%     \addlegendentry{SR};
%     \addlegendimage{blue}
%     \addlegendentry{GRR-CR};
%     \addlegendimage{black}
%     \addlegendentry{UE-CR};
%     \addlegendimage{brown}
%     \addlegendentry{SRR};
%     \addlegendimage{gray}
%     \addlegendentry{GRR};
%     \end{axis}
% \end{tikzpicture}
% 	\end{subfigure}\\
    \begin{subfigure}[b]{0.45\textwidth}
    \centering
		\begin{tikzpicture}
		\tikzstyle{every node}=[font=\small]
		\begin{axis}[
		width = 7cm,
		height = 4cm,
		axis lines = left,
		xmin = 0,
		xmax = 8,
		ymin = 0,
		ymax = 1,
		ylabel near ticks,
		xlabel = {$\varepsilon \  (n = 100)$}, 
		ylabel = $\frac{\recht{I}(X;Y)}{\recht{H}(X)}$]
		\addplot[color=red,mark=none] table[col sep = comma] {CSV/n_100/SR.csv};
		\addplot[color=blue,mark=none] table[col sep = comma] {CSV/n_100/GRRCR.csv};
		\addplot[color=black,mark=none] table[col sep = comma] {CSV/n_100/UECR.csv};
		\addplot[color=brown,mark=none] table[col sep = comma] {CSV/n_100/SRR.csv};
		\addplot[color=gray,mark=none] table[col sep = comma] {CSV/n_100/GRR.csv};
		\end{axis}
		\end{tikzpicture}
	\end{subfigure}
	\begin{subfigure}[b]{0.45\textwidth}
    \centering
		\begin{tikzpicture}
		\tikzstyle{every node}=[font=\small]
		\begin{axis}[
		width = 7cm,
		height = 4cm,
		axis lines = left,
		xmin = 0,
		xmax = 8,
		ymin = 0,
		ymax = 1,
		ylabel near ticks,
		xlabel = {$\varepsilon \  (n = 1000)$}, 
		ylabel = $\frac{\recht{I}(X;Y)}{\recht{H}(X)}$]
		\addplot[color=red,mark=none] table[col sep = comma] {CSV/n_1000/SR.csv};
		\addplot[color=blue,mark=none] table[col sep = comma] {CSV/n_1000/GRRCR.csv};
		\addplot[color=black,mark=none] table[col sep = comma] {CSV/n_1000/UECR.csv};
		\addplot[color=brown,mark=none] table[col sep = comma] {CSV/n_1000/SRR.csv};
		\addplot[color=gray,mark=none] table[col sep = comma] {CSV/n_1000/GRR.csv};
		\end{axis}
		\end{tikzpicture}
	\end{subfigure}
    \begin{subfigure}[b]{0.45\textwidth}
    \centering
		\begin{tikzpicture}
		\tikzstyle{every node}=[font=\small]
		\begin{axis}[
		width = 7cm,
		height = 4cm,
		axis lines = left,
		xmin = 0,
		xmax = 8,
		ymin = 0,
		ymax = 1,
		ylabel near ticks,
		xlabel = {$\varepsilon \  (n = 10000)$}, 
		ylabel = $\frac{\recht{I}(X;Y)}{\recht{H}(X)}$]
		\addplot[color=red,mark=none] table[col sep = comma] {CSV/n_10000/SR.csv};
		\addplot[color=blue,mark=none] table[col sep = comma] {CSV/n_10000/GRRCR.csv};
		\addplot[color=black,mark=none] table[col sep = comma] {CSV/n_10000/UECR.csv};
		\addplot[color=brown,mark=none] table[col sep = comma] {CSV/n_10000/SRR.csv};
		\addplot[color=gray,mark=none] table[col sep = comma] {CSV/n_10000/GRR.csv};
		\end{axis}
		\end{tikzpicture}
	\end{subfigure}
	\begin{subfigure}[b]{0.45\textwidth}
    \centering
		\begin{tikzpicture}
		\tikzstyle{every node}=[font=\small]
		\begin{axis}[
		width = 7cm,
		height = 4cm,
		axis lines = left,
		xmin = 0,
		xmax = 8,
		ymin = 0,
		ymax = 1,
		ylabel near ticks,
		xlabel = {$\varepsilon \  (n = 100000)$}, 
		ylabel = $\frac{\recht{I}(X;Y)}{\recht{H}(X)}$]
		\addplot[color=red,mark=none] table[col sep = comma] {CSV/n_100000/SR.csv};
		\addplot[color=blue,mark=none] table[col sep = comma] {CSV/n_10000/GRRCR.csv};
		\addplot[color=black,mark=none] table[col sep = comma] {CSV/n_10000/UECR.csv};
		\addplot[color=brown,mark=none] table[col sep = comma] {CSV/n_10000/SRR.csv};
		\addplot[color=gray,mark=none] table[col sep = comma] {CSV/n_10000/GRR.csv};
		\end{axis}
		\end{tikzpicture}
	\end{subfigure}
    \caption{\it Experiments on synthetic data with $n$ changing, $a_1 = a_2 = 5$ (\ref{pl:sr} IR, \ref{pl:grrcr} GRR-CR, \ref{pl:uecr}~\mbox{UE-CR}, \ref{pl:srr} SRR, \ref{pl:grr} GRR).}
    \label{fig:changen}
\end{figure}
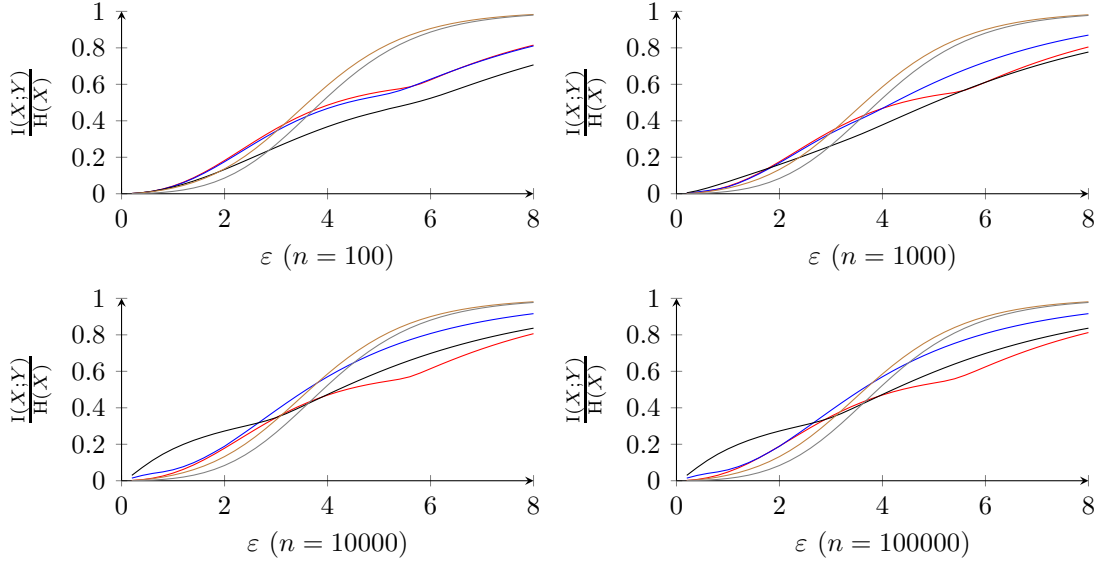

\begin{figure}
    \centering
%     \begin{subfigure}[b]{\textwidth}
%     \centering
% 		\begin{tikzpicture} 
% 		\tikzstyle{every node}=[font=\small]
%     \begin{axis}[%
%     width = 1mm,
%     height = 1mm,
%     hide axis,
%     xmin=0,
%     xmax=0.1,
%     ymin=0,
%     ymax=0.1,
%     legend columns=-1
%     ]
%     \addlegendimage{red}
%     \addlegendentry{$\alpha = 10^{-4}$};
%     \addlegendimage{blue}
%     \addlegendentry{$\alpha = 10^{-3}$};
%     \addlegendimage{black}
%     \addlegendentry{$\alpha = 10^{-2}$};
%     \addlegendimage{brown}
%     \addlegendentry{$\alpha = 10^{-1}$};
%     \end{axis}
% \end{tikzpicture}
% 	\end{subfigure}\\
    \begin{subfigure}[b]{0.45\textwidth}
    \centering
		\begin{tikzpicture}
		\tikzstyle{every node}=[font=\small]
		\begin{axis}[
		width = 7cm,
		height = 4cm,
		axis lines = left,
		xmin = 0,
		xmax = 8,
		ymin = 0,
		ymax = 1,
		ylabel near ticks,
		xlabel = {$\varepsilon \ $(GRR-CR)}, 
		ylabel = $\frac{\recht{I}(X;Y)}{\recht{H}(X)}$]
		\addplot[color=red,mark=none] table[col sep = comma] {CSV/alpha_GRRCR/1e-4.csv};
		\addplot[color=blue,mark=none] table[col sep = comma] {CSV/alpha_GRRCR/1e-3.csv};
		\addplot[color=black,mark=none] table[col sep = comma] {CSV/alpha_GRRCR/1e-2.csv};
		\addplot[color=brown,mark=none] table[col sep = comma] {CSV/alpha_GRRCR/1e-1.csv};
		\end{axis}
		\end{tikzpicture}
	\end{subfigure}
	\begin{subfigure}[b]{0.45\textwidth}
    \centering
		\begin{tikzpicture}
		\tikzstyle{every node}=[font=\small]
		\begin{axis}[
		width = 7cm,
		height = 4cm,
		axis lines = left,
		xmin = 0,
		xmax = 8,
		ymin = 0,
		ymax = 1,
		ylabel near ticks,
		xlabel = {$\varepsilon \ $(UE-CR)}, 
		ylabel = $\frac{\recht{I}(X;Y)}{\recht{H}(X)}$]
		\addplot[color=red,mark=none] table[col sep = comma] {CSV/alpha_UECR/1e-4.csv};
		\addplot[color=blue,mark=none] table[col sep = comma] {CSV/alpha_UECR/1e-3.csv};
		\addplot[color=black,mark=none] table[col sep = comma] {CSV/alpha_UECR/1e-2.csv};
		\addplot[color=brown,mark=none] table[col sep = comma] {CSV/alpha_UECR/1e-1.csv};
		\end{axis}
		\end{tikzpicture}
	\end{subfigure}
    \caption{\it Experiments on synthetic data with $\alpha$ changing, $a_1 = a_2 = 5$ (\ref{pl:sr} $\alpha = 10^{-4}$, \ref{pl:grrcr} $\alpha = 10^{-3}$, \ref{pl:uecr} $\alpha = 10^{-2}$, \ref{pl:srr} $\alpha = 10^{-1}$).}
    \label{fig:changealpha}
\end{figure}
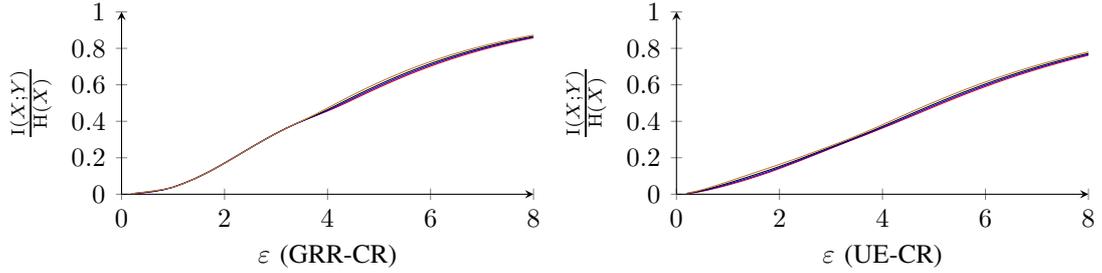

\subsection{Robustness of utility} \label{ssec:robust}

We also consider the robustness of the utility by comparing the `true utility' $\recht{I}_{P^{*}}(X;Y)$ to $\recht{I}_{\hat{P}}(X;Y)$, the latter of which is maximised in IR and CR. The results (for $a_1 = a_2 = 5$ ) are in Figure \ref{fig:robustuti}. As one can see, the true utility is on average often actually higher than the optimised utility, especially for small $n$. Furthermore, the difference between the two utilities rapidly becomes negligible for larger $n$. We conclude that IR and CR produce robust utility results.

\begin{figure}
    \centering
%     \begin{subfigure}[b]{\textwidth}
%     \centering
% 		\begin{tikzpicture} 
% 		\tikzstyle{every node}=[font=\small]
%     \begin{axis}[%
%     width = 1mm,
%     height = 1mm,
%     hide axis,
%     xmin=0,
%     xmax=0.1,
%     ymin=0,
%     ymax=0.1,
%     legend columns=-1
%     ]
%     \addlegendimage{red}
%     \addlegendentry{SR};
%     \addlegendimage{blue}
%     \addlegendentry{GRR-CR};
%     \addlegendimage{black}
%     \addlegendentry{UE-CR};
%     \end{axis}
% \end{tikzpicture}
% 	\end{subfigure}\\
    \begin{subfigure}[b]{0.45\textwidth}
    \centering
		\begin{tikzpicture}
		\tikzstyle{every node}=[font=\small]
		\begin{axis}[
		width = 7cm,
		height = 4cm,
		axis lines = left,
		xmin = 0,
		xmax = 8,
		ylabel near ticks,
		xlabel = {$\varepsilon \  (n = 100)$}, 
		ylabel = norm.uti.diff.]
		\addplot[color=red,mark=none] table[col sep = comma] {CSV/diff_100/SR.csv};
		\addplot[color=blue,mark=none] table[col sep = comma] {CSV/diff_100/GRRCR.csv};
		\addplot[color=black,mark=none] table[col sep = comma] {CSV/diff_100/UECR.csv};
		\end{axis}
		\end{tikzpicture}
	\end{subfigure}
	\begin{subfigure}[b]{0.45\textwidth}
    \centering
		\begin{tikzpicture}
		\tikzstyle{every node}=[font=\small]
		\begin{axis}[
		width = 7cm,
		height = 4cm,
		axis lines = left,
		xmin = 0,
		xmax = 8,
		ylabel near ticks,
		xlabel = {$\varepsilon \  (n = 1000)$}, 
		ylabel = norm.uti.diff.]
		\addplot[color=red,mark=none] table[col sep = comma] {CSV/diff_1000/SR.csv};
		\addplot[color=blue,mark=none] table[col sep = comma] {CSV/diff_1000/GRRCR.csv};
		\addplot[color=black,mark=none] table[col sep = comma] {CSV/diff_1000/UECR.csv};
		\end{axis}
		\end{tikzpicture}
	\end{subfigure}
    \begin{subfigure}[b]{0.45\textwidth}
    \centering
		\begin{tikzpicture}
		\tikzstyle{every node}=[font=\small]
		\begin{axis}[
		width = 7cm,
		height = 4cm,
		axis lines = left,
		xmin = 0,
		xmax = 8,
		ylabel near ticks,
		xlabel = {$\varepsilon \  (n = 10000)$}, 
		ylabel = norm.uti.diff.]
		\addplot[color=red,mark=none] table[col sep = comma] {CSV/diff_10000/SR.csv};
		\addplot[color=blue,mark=none] table[col sep = comma] {CSV/diff_10000/GRRCR.csv};
		\addplot[color=black,mark=none] table[col sep = comma] {CSV/diff_10000/UECR.csv};
		\end{axis}
		\end{tikzpicture}
	\end{subfigure}
	\begin{subfigure}[b]{0.45\textwidth}
    \centering
		\begin{tikzpicture}
		\tikzstyle{every node}=[font=\small]
		\begin{axis}[
		width = 7cm,
		height = 4cm,
		axis lines = left,
		xmin = 0,
		xmax = 8,
		ylabel near ticks,
		xlabel = {$\varepsilon \  (n = 100000)$}, 
		ylabel = norm.uti.diff.]
		\addplot[color=red,mark=none] table[col sep = comma] {CSV/diff_100000/SR.csv};
		\addplot[color=blue,mark=none] table[col sep = comma] {CSV/diff_100000/GRRCR.csv};
		\addplot[color=black,mark=none] table[col sep = comma] {CSV/diff_100000/UECR.csv};
		\end{axis}
		\end{tikzpicture}
	\end{subfigure}
    \caption{The average value of $\frac{\recht{I}_{\hat{P}}(X;Y)-\recht{I}_{P^{*}}(X;Y)}{\recht{H}(X)}$, for 200 randomly generated distributions, with $a_1 = a_2 = 5$ (\ref{pl:sr} IR, \ref{pl:grrcr} GRR-CR, \ref{pl:uecr} UE-CR).}
    \label{fig:robustuti}
\end{figure}
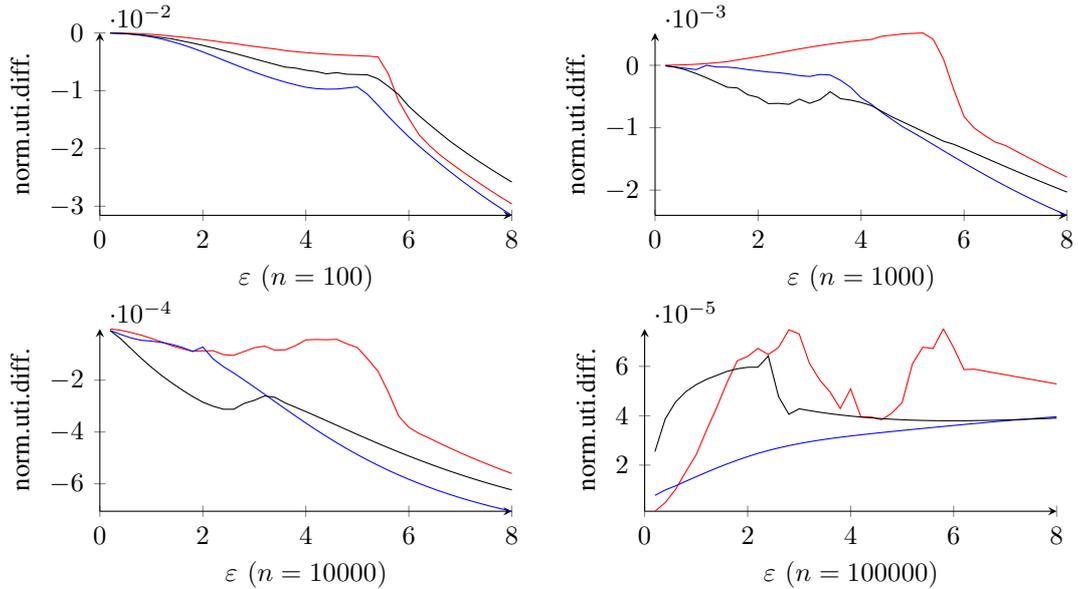

\subsection{Adult dataset} \label{ssec:adult}

We also perform numerical experiments on the adult-dataset ($n = 32561$) \cite{Dua:2019}, which contains demographic data from the 1994 US census. Some examples, where we use different categorical attributes from the dataset as $S$ and $U$, are depicted in Figure \ref{fig:adult}. To compare them to the synthetic data, we also perform experiments on synthetic data with the same $a_1,a_2$ as in the experiments in Figure \ref{fig:adult}; these exeperiments are in Figure \ref{fig:synthadult}. As we can see, the relative behaviour of the methods on the real data and the synthetic data of the same dimension align closely. The largest difference is the fact that IR outperforms GRR-CR for the synthetic data for $a_1 = 6$, $a_2 = 42$, but this is because the relative performance of IR and GRR-CR is distribution-specific, as we have seen in Section \ref{ssec:diff}. The close correspondence between the synthetic and real-data experiments lends additional validity to the experiments on synthetic data in the rest of this section.

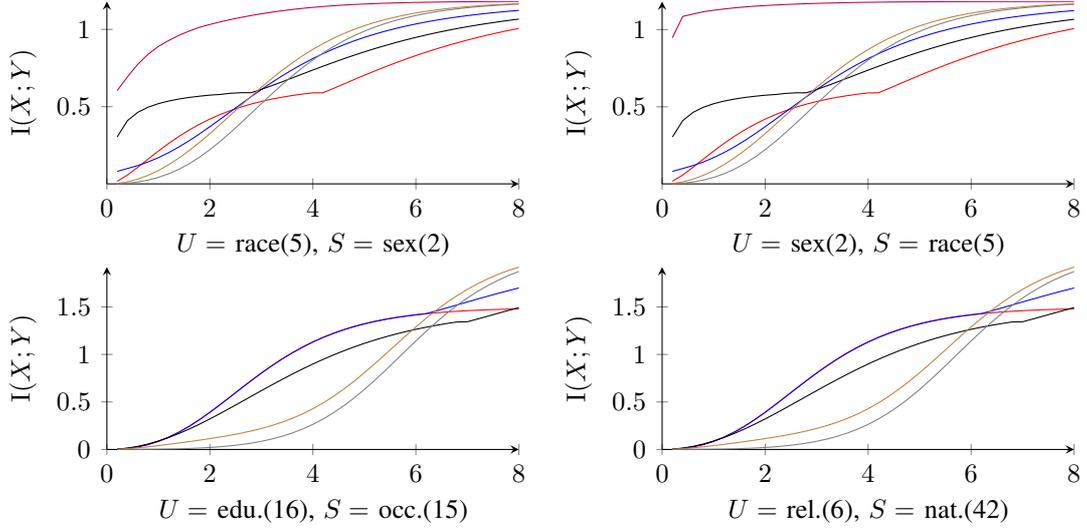
\begin{figure}
    \centering
%     \begin{subfigure}[b]{\textwidth}
%     \centering
% 		\begin{tikzpicture} 
% 		\tikzstyle{every node}=[font=\small]
%     \begin{axis}[%
%     width = 1mm,
%     height = 1mm,
%     hide axis,
%     xmin=0,
%     xmax=0.1,
%     ymin=0,
%     ymax=0.1,
%     legend columns=-1
%     ]
%     \addlegendimage{red}
%     \addlegendentry{SR};
%     \addlegendimage{blue}
%     \addlegendentry{GRR-CR};
%     \addlegendimage{black}
%     \addlegendentry{UE-CR};
%     \addlegendimage{brown}
%     \addlegendentry{SRR};
%     \addlegendimage{gray}
%     \addlegendentry{GRR};
%     \addlegendimage{purple}
%     \addlegendentry{PolyOpt};
%     \end{axis}
% \end{tikzpicture}
% 	\end{subfigure}\\
    \begin{subfigure}[b]{0.45\textwidth}
    \centering
		\begin{tikzpicture}
		\tikzstyle{every node}=[font=\small]
		\begin{axis}[
		width = 7cm,
		height = 4cm,
		axis lines = left,
		xmin = 0,
		xmax = 8,
		ylabel near ticks,
		xlabel = {$U=$ race(5), $S=$ sex(2)}, 
		ylabel = $\recht{I}(X;Y)$]
		\addplot[color=red,mark=none] table[col sep = comma] {CSV/S_race_U_sex/SR.csv};
		\addplot[color=blue,mark=none] table[col sep = comma] {CSV/S_race_U_sex/GRRCR.csv};
		\addplot[color=black,mark=none] table[col sep = comma] {CSV/S_race_U_sex/UECR.csv};
		\addplot[color=brown,mark=none] table[col sep = comma] {CSV/S_race_U_sex/SRR.csv};
		\addplot[color=gray,mark=none] table[col sep = comma] {CSV/S_race_U_sex/GRR.csv};
		\addplot[color=purple,mark=none] table[col sep = comma] {CSV/S_race_U_sex/PolyOpt.csv};
		\end{axis}
		\end{tikzpicture}
	\end{subfigure}
	\begin{subfigure}[b]{0.45\textwidth}
    \centering
		\begin{tikzpicture}
		\tikzstyle{every node}=[font=\small]
		\begin{axis}[
		width = 7cm,
		height = 4cm,
		axis lines = left,
		xmin = 0,
		xmax = 8,
		ylabel near ticks,
		xlabel = {$U=$ sex(2), $S=$ race(5)}, 
		ylabel = $\recht{I}(X;Y)$]
		\addplot[color=red,mark=none] table[col sep = comma] {CSV/S_sex_U_race/SR.csv};
		\addplot[color=blue,mark=none] table[col sep = comma] {CSV/S_sex_U_race/GRRCR.csv};
		\addplot[color=black,mark=none] table[col sep = comma] {CSV/S_sex_U_race/UECR.csv};
		\addplot[color=brown,mark=none] table[col sep = comma] {CSV/S_sex_U_race/SRR.csv};
		\addplot[color=gray,mark=none] table[col sep = comma] {CSV/S_sex_U_race/GRR.csv};
		\addplot[color=purple,mark=none] table[col sep = comma] {CSV/S_sex_U_race/PolyOpt.csv};
		\end{axis}
		\end{tikzpicture}
	\end{subfigure}
    \begin{subfigure}[b]{0.45\textwidth}
    \centering
		\begin{tikzpicture}
		\tikzstyle{every node}=[font=\small]
		\begin{axis}[
		width = 7cm,
		height = 4cm,
		axis lines = left,
		xmin = 0,
		xmax = 8,
		ylabel near ticks,
		xlabel = {$U=$ edu.(16), $S=$ occ.(15)}, 
		ylabel = $\recht{I}(X;Y)$]
		\addplot[color=red,mark=none] table[col sep = comma] {CSV/S_education_U_occupation/SR.csv};
		\addplot[color=blue,mark=none] table[col sep = comma] {CSV/S_education_U_occupation/GRRCR.csv};
		\addplot[color=black,mark=none] table[col sep = comma] {CSV/S_education_U_occupation/UECR.csv};
		\addplot[color=brown,mark=none] table[col sep = comma] {CSV/S_education_U_occupation/SRR.csv};
		\addplot[color=gray,mark=none] table[col sep = comma] {CSV/S_education_U_occupation/GRR.csv};
		\end{axis}
		\end{tikzpicture}
	\end{subfigure}
	\begin{subfigure}[b]{0.45\textwidth}
    \centering
		\begin{tikzpicture}
		\tikzstyle{every node}=[font=\small]
		\begin{axis}[
		width = 7cm,
		height = 4cm,
		axis lines = left,
		xmin = 0,
		xmax = 8,
		ylabel near ticks,
		xlabel = {$U=$ rel.(6), $S=$ nat.(42)}, 
		ylabel = $\recht{I}(X;Y)$]
		\addplot[color=red,mark=none] table[col sep = comma] {CSV/S_relationship_U_native-country/SR.csv};
		\addplot[color=blue,mark=none] table[col sep = comma] {CSV/S_relationship_U_native-country/GRRCR.csv};
		\addplot[color=black,mark=none] table[col sep = comma] {CSV/S_relationship_U_native-country/UECR.csv};
		\addplot[color=brown,mark=none] table[col sep = comma] {CSV/S_relationship_U_native-country/SRR.csv};
		\addplot[color=gray,mark=none] table[col sep = comma] {CSV/S_relationship_U_native-country/GRR.csv};
		\end{axis}
		\end{tikzpicture}
	\end{subfigure}
    \caption{Experiments on the categories \emph{sex}, \emph{race}, \emph{education}, \emph{occupation}, \emph{relation} and \emph{native-country} of the adult-dataset. Numbers between brackets indicate $a_1$ and $a_2$ (\ref{pl:PolyOpt} PolyOpt, \ref{pl:sr} IR, \ref{pl:grrcr} GRR-CR, \ref{pl:uecr} UE-CR, \ref{pl:srr} SRR, \ref{pl:grr} GRR).}
    \label{fig:adult}
\end{figure}

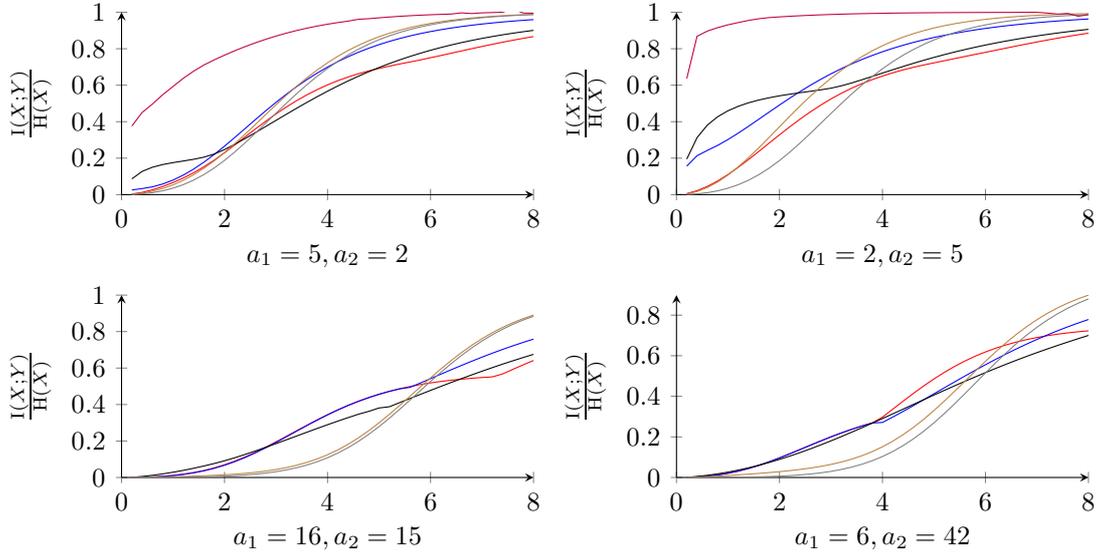
\begin{figure}
    \centering
%     \begin{subfigure}[b]{\textwidth}
%     \centering
%     \begin{tikzpicture} 
% 		\tikzstyle{every node}=[font=\small]
%     \begin{axis}[%
%     width = 1mm,
%     height = 1mm,
%     hide axis,
%     xmin=0,
%     xmax=0.1,
%     ymin=0,
%     ymax=0.1,
%     legend columns=-1
%     ]
%     \addlegendimage{red}
%     \addlegendentry{SR};
%     \addlegendimage{blue}
%     \addlegendentry{GRR-CR};
%     \addlegendimage{black}
%     \addlegendentry{UE-CR};
%     \addlegendimage{brown}
%     \addlegendentry{SRR};
%     \addlegendimage{gray}
%     \addlegendentry{GRR};
%     \addlegendimage{purple}
%     \addlegendentry{PolyOpt};
%     \end{axis}
% \end{tikzpicture}
% 	\end{subfigure}\\
    \begin{subfigure}[b]{0.45\textwidth}
    \centering
		\begin{tikzpicture}
		\tikzstyle{every node}=[font=\small]
		\begin{axis}[
		width = 7cm,
		height = 4cm,
		axis lines = left,
		xmin = 0,
		xmax = 8,
		ymin = 0,
		ymax = 1,
		ylabel near ticks,
		xlabel = {$a_1 = 5, a_2 = 2$}, 
		ylabel = $\frac{\recht{I}(X;Y)}{\recht{H}(X)}$]
		\addplot[color=red,mark=none] table[col sep = comma] {CSV/adult_S5U2/SR.csv};
		\addplot[color=blue,mark=none] table[col sep = comma] {CSV/adult_S5U2/GRRCR.csv};
		\addplot[color=black,mark=none] table[col sep = comma] {CSV/adult_S5U2/UECR.csv};
		\addplot[color=brown,mark=none] table[col sep = comma] {CSV/adult_S5U2/SRR.csv};
		\addplot[color=gray,mark=none] table[col sep = comma] {CSV/adult_S5U2/GRR.csv};
		\addplot[color=purple,mark=none] table[col sep = comma] {CSV/adult_S5U2/PolyOpt.csv};
		\end{axis}
		\end{tikzpicture}
	\end{subfigure}
	\begin{subfigure}[b]{0.45\textwidth}
    \centering
		\begin{tikzpicture}
		\tikzstyle{every node}=[font=\small]
		\begin{axis}[
		width = 7cm,
		height = 4cm,
		axis lines = left,
		xmin = 0,
		xmax = 8,
		ymin = 0,
		ymax = 1,
		ylabel near ticks,
		xlabel = {$a_1 = 2, a_2 = 5$}, 
		ylabel = $\frac{\recht{I}(X;Y)}{\recht{H}(X)}$]
		\addplot[color=red,mark=none] table[col sep = comma] {CSV/adult_S2U5/SR.csv};
		\addplot[color=blue,mark=none] table[col sep = comma] {CSV/adult_S2U5/GRRCR.csv};
		\addplot[color=black,mark=none] table[col sep = comma] {CSV/adult_S2U5/UECR.csv};
		\addplot[color=brown,mark=none] table[col sep = comma] {CSV/adult_S2U5/SRR.csv};
		\addplot[color=gray,mark=none] table[col sep = comma] {CSV/adult_S2U5/GRR.csv};
		\addplot[color=purple,mark=none] table[col sep = comma] {CSV/adult_S2U5/PolyOpt.csv};
		\end{axis}
		\end{tikzpicture}
	\end{subfigure}
    \begin{subfigure}[b]{0.45\textwidth}
    \centering
		\begin{tikzpicture}
		\tikzstyle{every node}=[font=\small]
		\begin{axis}[
		width = 7cm,
		height = 4cm,
		axis lines = left,
		xmin = 0,
		xmax = 8,
		ymin = 0,
		ymax = 1,
		ylabel near ticks,
		xlabel = {$a_1 = 16, a_2 = 15$}, 
		ylabel = $\frac{\recht{I}(X;Y)}{\recht{H}(X)}$]
		\addplot[color=red,mark=none] table[col sep = comma] {CSV/adult_S16U15/SR.csv};
		\addplot[color=blue,mark=none] table[col sep = comma] {CSV/adult_S16U15/GRRCR.csv};
		\addplot[color=black,mark=none] table[col sep = comma] {CSV/adult_S16U15/UECR.csv};
		\addplot[color=brown,mark=none] table[col sep = comma] {CSV/adult_S16U15/SRR.csv};
		\addplot[color=gray,mark=none] table[col sep = comma] {CSV/adult_S16U15/GRR.csv};
		\end{axis}
		\end{tikzpicture}
	\end{subfigure}
	\begin{subfigure}[b]{0.45\textwidth}
    \centering
		\begin{tikzpicture}
		\tikzstyle{every node}=[font=\small]
		\begin{axis}[
		width = 7cm,
		height = 4cm,
		axis lines = left,
		xmin = 0,
		xmax = 8,
		ylabel near ticks,
		xlabel = {$a_1 = 6, a_2 = 42$}, 
		ylabel = $\frac{\recht{I}(X;Y)}{\recht{H}(X)}$]
		\addplot[color=red,mark=none] table[col sep = comma] {CSV/adult_S6U42/SR.csv};
		\addplot[color=blue,mark=none] table[col sep = comma] {CSV/adult_S6U42/GRRCR.csv};
		\addplot[color=black,mark=none] table[col sep = comma] {CSV/adult_S6U42/UECR.csv};
		\addplot[color=brown,mark=none] table[col sep = comma] {CSV/adult_S6U42/SRR.csv};
		\addplot[color=gray,mark=none] table[col sep = comma] {CSV/adult_S6U42/GRR.csv};
		\end{axis}
		\end{tikzpicture}
	\end{subfigure}
    \caption{\it Experiments on synthetic data, with $a_1,a_2,n$ as in the adult-dataset. Horizontal axis is $\varepsilon$-LDP; vertical axis is $\recht{I}(X;Y)/\recht{H}(X)$ (\ref{pl:PolyOpt} PolyOpt, \ref{pl:sr} IR, \ref{pl:grrcr} GRR-CR, \ref{pl:uecr} UE-CR, \ref{pl:srr} SRR, \ref{pl:grr} GRR).}
    \label{fig:synthadult}
\end{figure}

\section{Conclusion and future work} \label{sec:disc}

In this paper, we presented a number of algorithms that, given a desired privacy level $\varepsilon$, an estimated distribution $\hat{P}$, and a set of probability distributions $\mathcal{F}$ of a specified form, return a release protocol that aim to maximise the mutual information between input and output, while satisfying privacy w.r.t. a given sensitive part of the data, for all distributions in $\mathcal{F}$. In the case that $\mathcal{F} = \mathcal{P}_{\mathcal{X}}$, we have introduced SRR, which we have shown to be optimal for this $\mathcal{F}$ in the low privacy regime, irrespective of the actual probability distribution. Furthermore, experiments show that in the low privacy regime SRR outperforms most of our other algorithms, even though these have smaller $\mathcal{F}$.  The privacy level at which SRR overtakes the other algorithms in utility is lower for larger input spaces and smaller $\mathcal{F}$. However, in the high privacy regime the other algorithms offer significantly better utility. This shows the validity of using confidence sets in the RLDP framework.

In the case that $\mathcal{F}$ is a confidence set around $\hat{P}$, arising from a $\chi^2$-test with given confidence level, we offer multiple algorithms. One of these, PolyOpt, offers significantly higher utility, especially in the high privacy regime. However, it relies on vertex enumeration, making it computationally infeasible for larger input spaces. The other 3 algorithms, SR, GRR-CR and UE-CR, rely on processing the sensitive and non-sensitive data separately. These algorithms rely on low-dimensional optimisation, independent of the size of the input space, allowing these to be used when PolyOpt is outside the computational capabilities. Of these protocols, UE-CR is the best option when either $\mathcal{F}$ or the input space is small. SR and GRR-CR perform similar in the high privacy regime, with GRR-CR performing better for input distributions with large probability.

Our results suggest several avenues for future research. First, one may want to incorporate not only robustness in privacy, but also in utility, i.e. to find the protocol $\mathcal{Q}$ that maximises $\min_{P \in \mathcal{F}} \recht{I}_P(X;Y)$. An obstacle for this is that $\recht{I}_P(X;Y)$ is concave in $P$, which makes finding its minimum over $\mathcal{F}$ difficult. Second, instead of looking at the situation where $X$ splits into a sensitive part $S$ and a non-sensitive part $U$, one can consider the more general case that $X$ is correlated with the sensitive data $S$. This is already done in work on the privacy funnel, but this generally does not incorporate robustness. Furthermore, the utility of IR and CR might be improved in the high privacy regume by incorporating other LDP protocols than GRR. It is shown in \cite{kairouz2014extremal} that GRR is the optimal LDP protocol for high $\varepsilon$, but for low $\varepsilon$ the optimum typically takes a different form. One obstacle in incorporating this is that these optima depend on $P^*$, which is inaccessible in the RLDP framework.

\section*{Acknowledgements}

This   work   was   supported   by   NWO   grant   628.001.026.

\printbibliography

\appendices

\section{Proofs} \label{sec:app}

\subsection{Proof of Theorem \ref{thm:sgrr}} \label{ssec:pfsgrr}

We follow the proof of Theorem 14 in \cite{kairouz2014extremal}. For $C \in \mathbb{R}^{\mathcal{X}}_{\geq 0}$, define 
\begin{equation}
\mu(C) = \sum_x P_xC_x \log \frac{C_x}{\sum_{x'}P_{x'}C_{x'}}.
\end{equation}
Then the utility of a protocol $\mathcal{Q}\colon \mathcal{X} \rightarrow \mathcal{Y}$ is given by $\recht{I}_P(X;Y)=\sum_y \mu(Q_{y|\bullet})$. Furthermore, $\mu$ is a sublinear function in the sense of \cite[Definition 1]{kairouz2014extremal}.

We fix an $\varepsilon > 0$. Furthermore, let $\mathcal{C} \subset \mathbb{R}_{\geq 0}^{\mathcal{X}}$ be the positive cone defined by the inequalities of the following form, for $s,s' \in \mathcal{S}$ with $s \neq s'$ and $u,u' \in \mathcal{U}$:
\begin{equation}
C_{s,u} \leq \textrm{e}^{\varepsilon}C_{s',u'}.
\end{equation}
Then a protocol $\mathcal{Q}$ satisfies $\varepsilon$-SLDP if and only if each $Q_{y|\bullet}$ is an element of $\mathcal{C}$. Furthermore, $\mathcal{C}$ is spanned (as a cone) by the set 
\begin{equation}
\mathcal{V} = \left\{v \in \mathcal{R}^{\mathcal{X}}:\substack{\forall x: v_x \in \{1,\textrm{e}^{\varepsilon}\},\\ |\{s:\exists u \textrm{ s.t. }| \  v_{s,u} = \textrm{e}^{\varepsilon}\} \geq 2}\right\} \cup \left\{v \in \mathcal{R}^{\mathcal{X}}:\exists s \textrm{ s.t.} \substack{\forall u: v_{s,u} \in \{\textrm{e}^{-\varepsilon},\textrm{e}^{\varepsilon}\};\\\forall s'\neq s, \forall u: v_{s',u} = 1}\right\}.
\end{equation}
Let $\mathcal{D}$ be the polytope spanned by $\mathcal{V}$. If $\mathcal{Q}$ satisfies $\varepsilon$-SLDP, then every column $Q_{y|\bullet}$ is of the form $\theta_y \cdot d_y$, where $d_y \in \mathcal{D}$ and $\theta_y \in \mathbb{R}_{\geq 0}$ are such that $\sum_y \theta_y d_y = 1_{\mathcal{X}}$. Analogous to the proof of Theorems 2 and 4 in \cite[Section 7]{kairouz2014extremal}, one proves that the optimal $\mathcal{Q}$ is found by taking $b = a$, and taking $d_y \in \mathcal{V}$ for all $d$. Since
\begin{equation}
\recht{I}(X;Y) = \sum_y \mu(Q_{y|\bullet}) = \sum_y \theta_y \mu(d_y)
\end{equation}
we can find the optimal $\mathcal{Q}$ by solving the following optimisation problem, where $m$ is the vector $(\mu(v))_{v \in \mathcal{V}}$, and where $A \in \mathbb{R}^{\mathcal{X}\times\mathcal{V}}$ is the matrix whose $v$-th column is $v$:
\begin{align*}
\textrm{maximise}_{\theta \in \mathbb{R}^{\mathcal{V}}} \ & m\cdot \theta \\
\textrm{such that} \ & A\cdot \theta = 1_{\mathcal{X}},\\
 & \theta \geq 0.
\end{align*}
From here, we follow \cite[Section 9.5]{kairouz2014extremal}. The dual to the above problem is
\begin{align*}
\textrm{minimise}_{\alpha \in \mathbb{R}^{\mathcal{X}}} \ & (1_{\mathcal{X}}) \cdot \alpha \\
\textrm{such that} \ & A^{\recht{T}} \cdot \alpha \geq m,\\
 & \alpha \geq 0.
\end{align*}

By duality we have $\max_{\theta} m \cdot \theta = \min_\alpha (1_{\mathcal{X}}) \cdot \alpha$. We describe $\alpha^*$ and $\theta^*$, depending on $\varepsilon$, such that for $\varepsilon$ large enough one has $A^{\recht{T}} \cdot \alpha^* \geq m$, such that $m \cdot \theta^* = (1_{\mathcal{X}})\cdot \alpha^*$ and $A\theta^* = 1_{\mathcal{X}}$, and such that $\theta^*$ corresponds to SGRR, i.e. for each $y \in \mathcal{Y} = \mathcal{X}$ there is a $\hat{v}_y \in \mathcal{V}$ such that $Q_{y|\bullet} = \theta^*_y v_y$. Together, this proves that SGRR is optimal for $\varepsilon \gg 0$.

More concretely, for $y = (s,u) \in \mathcal{X}$, define $\hat{v}_y$ by
\begin{equation}
(\hat{v}_{y})_{s',u'} = \left\{\begin{array}{ll}
\textrm{e}^{\varepsilon},& \ \textrm{if $(s',u') = (s,u)$},\\
\textrm{e}^{-\varepsilon},& \ \textrm{if $s'= s$ and $u' \neq u$},\\
1,& \ \textrm{if $s' \neq s$},
\end{array}\right.
\end{equation}
and let $\theta^* \in \mathbb{R}^{\mathcal{V}}$ be given by
\begin{equation}
\theta^*_v = \left\{\begin{array}{ll}
\frac{1}{\textrm{e}^{\varepsilon}+\textrm{e}^{-\varepsilon}(a_2-1)+a-a_2},& \ \textrm{if there is a $y \in \mathcal{X}$ such that $v = \hat{v}_y$,}\\
0,& \ \textrm{otherwise;}
\end{array}\right.
\end{equation}
Then SRR satisfies $Q_{y|\bullet} = \theta^*_{\hat{v}_y} \hat{v}_y$ for all $y \in \mathcal{X}$, and also
\begin{align}
(A\theta^*)_x &= \sum_v A_{x,v} \theta^*_v \\
&= \sum_v v_x \theta^*_v \\
&= \frac{\sum_y (\hat{v}_y)_x}{\textrm{e}^{\varepsilon}+\textrm{e}^{-\varepsilon}(a_2-1)+a-a_2} \\
&= 1,
\end{align}
which shows that $A\theta^* = 1_{\mathcal{X}}$. Furthermore, define $\alpha^* \in \mathbb{R}^{\mathcal{X}}$ by
\begin{equation}
\alpha^*_{s,u} = c_1\mu(\hat{v}_{s,u}) + c_2\sum_{u' \neq u} \mu(\hat{v}_{s,u'}) + c_3\sum_{\substack{s' \neq s,\\u'}}\mu(\hat{v}_{s',u'}),
\end{equation}
where
\begin{align}
c_1 &= \tfrac{-(a_2-2)(a_2-1)+(a-a_2+1)(a_2-2)\textrm{e}^{\varepsilon}+(a-2a_2+1)\textrm{e}^{2\varepsilon}+\textrm{e}^{3\varepsilon}}{(\textrm{e}^{\varepsilon}-1)(\textrm{e}^{\varepsilon}+1)(\textrm{e}^{\varepsilon}-a_2+1)(\textrm{e}^{\varepsilon}+(a_2-1)\textrm{e}^{-\varepsilon}+a-a_2)},\\
c_2 &= \tfrac{a_2-1+(a-a_2+1)\textrm{e}^{\varepsilon}}{(\textrm{e}^{\varepsilon}-1)(\textrm{e}^{\varepsilon}+1)(\textrm{e}^{\varepsilon}-a_2+1)(\textrm{e}^{\varepsilon}+(a_2-1)\textrm{e}^{-\varepsilon}+a-a_2)},\\
c_3 &= \tfrac{-\textrm{e}^{2\varepsilon}}{(\textrm{e}^{\varepsilon}-1)(\textrm{e}^{\varepsilon}-a_2+1)(\textrm{e}^{\varepsilon}+(a_2-1)\textrm{e}^{-\varepsilon}+a-a_2)}.
\end{align}
One readily calculates that for all $x$ we have
\begin{align} m \cdot \theta^* = (1_{\mathcal{X}})\cdot \alpha^* &= \tfrac{1}{\textrm{e}^{\varepsilon}+\textrm{e}^{-\varepsilon}(a_2-1)+a-a_2}\sum_x \mu(\hat{v}_x),\\
\hat{v}_x \cdot \alpha^* = m_{\hat{v}_x} &= \mu(\hat{v}_x). \label{eq:maxf1}
\end{align}

It remains to be shown that $\alpha^*$ satisfies the dual problem for $\varepsilon \gg 0$, i.e. $A^{\recht{T}}\alpha \geq m$ for $\varepsilon$ large enough. To this end, for $v \in \mathcal{V}$, set 
\begin{align}
F_v &= \{x \in \mathcal{X}:v_x = \textrm{e}^{\varepsilon}\},\\
G_v &= \{x \in \mathcal{X}:v_x = 1\},\\
H_v &= \{x \in \mathcal{X}:v_x = \textrm{e}^{-\varepsilon}\},
\end{align}
Then $\#F_v \geq 1$ for all $v$, and $\#F_v = 1$ if and only if there exist $s,u$ such that $v = \hat{v}_{s,u}$. We write $P_{F_v} = \sum_{x \in F_v} P_x$ and likewise for $G_v$, $H_v$. For large $\varepsilon$ we have
\begin{align}
m_v = \mu(v) &= \textrm{e}^{\varepsilon}\sum_{x \in F_v} P_x \log \frac{1}{P_{F_v}+\textrm{e}^{-\varepsilon}P_{G_v}+\textrm{e}^{-2\varepsilon}P_{H_v}}\\
& \ \ +\sum_{x \in G_v} P_x \log \frac{1}{\textrm{e}^{\varepsilon}P_{F_v}+P_{G_v}+\textrm{e}^{-\varepsilon}P_{H_v}}\\
& \ \ +\textrm{e}^{-\varepsilon}\sum_{x \in H_x} P_x \log \frac{1}{\textrm{e}^{2\varepsilon}P_{F_v}+\textrm{e}^{\varepsilon}P_{G_v}+P_{H_v}} \\
&= \left(-P_{F_v} \log P_{F_v}\right)\textrm{e}^{\varepsilon}+\mathcal{O}(\varepsilon)
\end{align}
and furthermore
\begin{align}
c_1 &= \textrm{e}^{-\varepsilon}+\mathcal{O}(\textrm{e}^{-2\varepsilon}),\\
c_2,c_3 &= \mathcal{O}(\textrm{e}^{-2\varepsilon}),\\
\alpha^*_x &= c_1\mu(\hat{v}_x) + (c_2+c_3)\mathcal{O}(\textrm{e}^{\varepsilon})\\
&= -P_x\log P_x + \mathcal{O}(\varepsilon\textrm{e}^{-\varepsilon}),\\
v^{\recht{T}}\alpha^* &= \left(-\sum_{x \in F_v} P_x \log P_x\right)\textrm{e}^{\varepsilon}+\mathcal{O}(\varepsilon).
\end{align}

For $|F_v|\geq 2$ one has $P_{F_v} \log P_{F_v} > \sum_{x \in F_v} P_x \log P_x$. This means that if $v$ is not of the form $\hat{v}_x$, one has $v^{\recht{T}}\alpha^* \geq m_v$ for $\varepsilon$ large enough. Together with (\ref{eq:maxf1}) this shows that $A^{\recht{T}}\alpha^* \geq m$ for $\varepsilon$ large enough; this concludes the proof.

\subsection{Proof of Proposition \ref{prop:l1bound1}} \label{ssec:pfl1bound1}

Before we can proof this proposition, we need the following auxiliary lemma.

\begin{lemma} \label{lem:noninc}
Let $B \in \mathbb{R}_{\geq 1}$. Then the function $g\colon[0,1] \rightarrow \mathbb{R}$ given by
\begin{equation}
g(x) = \frac{B(1-2x)+\sqrt{B(B+4x(1-x))}}{B+1}
\end{equation}
is nonincreasing.
\end{lemma}

\begin{proof}
Since $B \geq 1$ we have for all $x \in [0,1]$ that $B^2-1 + 4(B+1)x(1-x) \geq 0$. Rearranging terms, it follows that
\begin{equation}
B(B+4x(1-x)) \geq 1-4x+4x^2,
\end{equation}
hence $\sqrt{B(B+4x(1-x))} \geq 1-2x$. Using this, one calculates
\begin{equation}
g'(x) = \frac{2B\left(1-2x-\sqrt{B(B+4x(1-x)}\right)}{(B+1)\sqrt{B(B+4x(1-x))}} \leq 0,
\end{equation}
hence $g$ is nonincreasing.
\end{proof}

\begin{proof}[Proof of Proposition \ref{prop:l1bound1}]
The distribution $P$ maximising $||P-\hat{P}||_1$ is located on the boundary of $\mathcal{F}$, hence $\sum_x \frac{\tilde{P}_x^2}{P_x} = B+1$. We define sets
\begin{align}
\mathcal{X}_1 &:= \left\{x \in \mathcal{X}: P_x \geq \hat{P}_x > 0\right\}, \\
\mathcal{X}_2 &:= \left\{x \in \mathcal{X}: P_x < \hat{P}_x\right\}, \\
\mathcal{X}_3 &:= \left\{x \in \mathcal{X}: \hat{P}_x = 0\right\}. 
\end{align}
Note that $\mathcal{X}_2$ and $\mathcal{X}_1 \cup \mathcal{X}_3$ both have to be nonempty. Then
\begin{equation}
||P-\hat{P}||_1 = \sum_{x \in \mathcal{X}_1} (P_x-\hat{P}_x) + \sum_{x \in \mathcal{X}_2} (\hat{P}_x-P_x) + \sum_{x \in \mathcal{X}_3} P_x.
\end{equation}
We can find the $P$ maximising this, subject to the constraints $\sum_x \frac{\hat{P}_x^2}{P_x} = B+1$ and $\sum_x P_x = 1$, by finding critical points of the Lagrange multiplier expression
\begin{equation}
\sum_{x \in \mathcal{X}_1} (P_x-\hat{P}_x) + \sum_{x \in \mathcal{X}_2} (\hat{P}_x-P_x) + \sum_{x \in \mathcal{X}_3} P_x + \lambda\left(\sum_x \frac{\hat{P}_x^2}{P_x} - B-1\right)+\mu\left(\sum_x P_x-1\right).
\end{equation}
Differentiating this with respect to $P_x$ for $x \in \mathcal{X}_1,\mathcal{X}_2,\mathcal{X}_3$, respectively, we get
\begin{align}
\forall x \in \mathcal{X}_1: \ \ 1-\frac{\lambda\hat{P}_x^2}{P_x^2}+\mu &= 0,\label{eq:x1}\\
\forall x \in \mathcal{X}_2: \ \ -1-\frac{\lambda\hat{P}_x^2}{P_x^2}+\mu &= 0,\label{eq:x2}\\
\forall x \in \mathcal{X}_3: \ \ 1+\mu &= 0.\label{eq:x3}
\end{align}
If $\mathcal{X}_1,\mathcal{X}_2,\mathcal{X}_3$ are all nonempty, then (\ref{eq:x3}) implies $\mu = -1$, so from (\ref{eq:x1}) we get $\lambda = 0$; however, then (\ref{eq:x2}) leads to a contradiction. Hence either $\mathcal{X}_1$ or $\mathcal{X}_3$ is empty; we will discuss these cases separately.

Suppose $\mathcal{X}_1 = \varnothing$; then $\sum_{x \in \mathcal{X}_2} \hat{P}_x = 1$. From (\ref{eq:x2}) and (\ref{eq:x3}) we find $P_x = \sqrt{-\frac{\lambda}{2}}\hat{P}_x$ for all $x \in \mathcal{X}_2$. Writing $c_- := \sqrt{-\frac{\lambda}{2}}$, we know that $c_-$ should satisfy
\begin{equation}
\frac{1}{c_-} = \sum_{x \in \mathcal{X}_2} \frac{\hat{P}_x}{c_-} = \sum_{x \in \mathcal{X}_2} \frac{\hat{P}_x^2}{P_x} = B+1.
\end{equation}
Hence $c_- = \frac{1}{B+1}$, and
\begin{align}
||P-\hat{P}||_1 &= 2\sum_{x \in \mathcal{X}_2} (\hat{P}_x-P_x) \\
&= 2(1-c_-)\sum_{x \in \mathcal{X}_2} \hat{P}_x \\
&= \frac{2B}{B+1},
\end{align}
which is indeed the formula in (\ref{eq:l1sup1}) when $P_{x_{\recht{min}}} = 0$. Furthermore, by the AM-GM inequality we have $\frac{2B}{B+1} = \sqrt{B}\frac{2}{\sqrt{B}+\frac{1}{\sqrt{B}}} \leq \sqrt{B}$, which shows that $||P-\hat{P}||_1 \leq \sqrt{B}$ in general, and in particular for $B \leq 1$.

Now suppose $\mathcal{X}_3 = \varnothing$. In that case we find from (\ref{eq:x1}) that for $x \in \mathcal{X}_1$ we have $P_x = \sqrt{\frac{\lambda}{\mu+1}}\hat{P}_x$, while for $x \in \mathcal{X}_2$ we have $P_x = \sqrt{\frac{\lambda}{\mu-1}}\hat{P}_x$. Setting $c_+:= \sqrt{\frac{\lambda}{\mu+1}}$ and $c_- := \sqrt{\frac{\lambda}{\mu-1}}$, then
\begin{align}
\hat{P}_{\mathcal{X}_1}c_+ + (1-\hat{P}_{\mathcal{X}_1})c_- &= \sum_{x \in \mathcal{X}_1} c_+\hat{P}_x + \sum_{x \in \mathcal{X}_2} c_-\hat{P}_x\\
&= 1, \label{eq:l1bound1}\\
\frac{\hat{P}_{\mathcal{X}_1}}{c_+} + \frac{1-\hat{P}_{\mathcal{X}_1}}{c_-} &= \sum_{x \in \mathcal{X}_1} \frac{\hat{P}_x}{c_+} + \sum_{x \in \mathcal{X}_2} \frac{\hat{P}_x}{c_-} \\
&= \sum_{x \in \mathcal{X}_1} \frac{\hat{P}_x^2}{P_x} + \sum_{x \in \mathcal{X}_2} \frac{\hat{P}_x^2}{P_x} \\
&= B+1. \label{eq:l1bound2}
\end{align}
Jointly solving (\ref{eq:l1bound1}) and (\ref{eq:l1bound2}) we find
\begin{align}
c_+ &= \frac{B+2\hat{P}_{\mathcal{X}_1}\pm\sqrt{B^2+4B\hat{P}_{\mathcal{X}_1}-4B\hat{P}_{\mathcal{X}_1}^2}}{2(B+1)\hat{P}_{\mathcal{X}_1}},\\
c_- &= \frac{B+2(1-\hat{P}_{\mathcal{X}_1})\mp\sqrt{B^2+4B\hat{P}_{\mathcal{X}_1}-4B\hat{P}_{\mathcal{X}_1}^2}}{2(B+1)(1-\hat{P}_{\mathcal{X}_1})}.
\end{align}
By definition of $\mathcal{X}_1$ and $\mathcal{X}_2$ we know that $c_+ \geq 1$ and $c_- < 1$; this only occurs if $c_+$ is the "$+$" solution while $c_-$ is the "-" solution. It follows that
\begin{align}
||P-\hat{P}||_1 &= \hat{P}_{\mathcal{X}_1}(c_+-1) + (1-\hat{P}_{\mathcal{X}_1})(1-c_-) \\
&= \frac{B-2B\hat{P}_{\mathcal{X}_1}+\sqrt{B^2+4B\hat{P}_{\mathcal{X}_1}-4B\hat{P}_{\mathcal{X}_1}^2}}{B+1}. \label{eq:l1max}
\end{align}
It follows that the $P$ maximising $||P-\hat{P}||_1$ is obtained by finding the (nonempty) subset $\mathcal{X}_1 \subset \mathcal{X}$ that maximises (\ref{eq:l1max}). By Lemma \ref{lem:noninc}, this is when $\mathcal{X}_1 = \{x_{\recht{min}}\}$ if $B \geq 1$, proving (\ref{eq:l1sup1}). If $B < 1$, the optimal $\mathcal{X}_1$ is unfortunately harder to determine. However, we can still find an upper bound, by finding the value of $\hat{P}_{\mathcal{X}_1}$ that maximises (\ref{eq:l1max}). We find the maximum by taking the derivative with respect to $\hat{P}_{\mathcal{X}_1}$, and we have to solve
\begin{equation}
\frac{-2B}{B+1} + \frac{2B-4B\hat{P}_{\mathcal{X}_1}}{(B+1)\sqrt{B^2+4B\hat{P}_{\mathcal{X}_1}-4B\hat{P}_{\mathcal{X}_1}^2}} = 0,
\end{equation}
which leads to $\hat{P}_{\mathcal{X}_1} = \frac{1-\sqrt{B}}{2}$. Substituting this in (\ref{eq:l1max}), we find
\begin{align}
||P-\hat{P}_1|| &\leq \frac{B-B(1-\sqrt{B})+\sqrt{B^2+2B(1-\sqrt{B})-B(1-\sqrt{B})^2}}{B+1} \\
&= \sqrt{B}. \qedhere
\end{align}
\end{proof}

\subsection{Proof of Theorem \ref{thm:poly2}} \label{ssec:pfpoly}

This is essentially analogous to the proof of Theorem 4 in \cite{kairouz2014extremal}; the main difference is that the equivalent of $\hat{\Gamma}$ is a hypercube, so there a vertex enumeration step is not needed. Let $\mathcal{Q}$ be a protocol such that $Q_y \in \Gamma$ for all $y$; then there exist $\alpha_y \in \mathbb{R}_{\geq 0}$, $\gamma_y \in \hat{\Gamma}$ such that $Q_y = \alpha_y\gamma_y$. One has
\begin{equation}
\recht{I}_{\hat{P}}(X;Y) = \sum_y \mu^1(Q_y) = \sum_y \alpha_y \mu^1(\gamma_y).
\end{equation}
Since $\hat{\Gamma}$ is the convex hull of $\mathcal{V}$, we can write $\gamma_y = \sum_v \lambda_{y,v}v$ for suitable constants $\lambda_{y,v}$. Define $\theta \in \mathbb{R}_{\geq 0}^{\mathcal{V}}$ by $\theta_v = \sum_y \lambda_{y,v}\alpha_y$. Then
\begin{equation}
    \sum_{v} \theta_v v = \sum_y Q_y = 1_{\mathcal{X}}.
\end{equation}
As such, the matrix $Q' \in \mathbb{R}^{\mathcal{V}\times\mathcal{X}}$ defined by $Q'_v = \theta_v v$ defines a privacy protocol $\mathcal{Q}'$. One has
\begin{align}
    \recht{I}_{\hat{P}}(X;\mathcal{Q}'(X)) &= \sum_v \mu^1(Q'_v) \\
    &= \sum_v \theta_v \mu^1(v) \\
    &= \sum_y \alpha_y \sum_v \lambda_{y,v} \mu^1(v) \\
    &\geq \sum_y \alpha_y  \mu^1\left(\sum_v \lambda_{y,v}\right) \\
    &= \recht{I}_{\hat{P}}(X;\mathcal{Q}(X)),
\end{align}
where we use the fact that $\mu^1$ is convex. This shows that the $Q_y$ of the optimal protocol satisfying Theorem \ref{thm:poly} are all of the form $\theta_v \cdot v$; hence (\ref{eq:polyprob}) yields the optimal protocol. For $\mathcal{Q}^2$, note that
\begin{align}
    \inf_P (X;\mathcal{Q}^2(X)) &= \inf_P \sum_v \hat{\theta}^2_v \sum_x v_x P_x \log \frac{v_x}{\sum_{x'}v_{x'}P_{x'}} \\
    &\geq \sum_v \hat{\theta}^2_v \inf_P \sum_x v_x P_x \log \frac{v_x}{\sum_{x'}v_{x'}P_{x'}} \\
    &= \sum_v \hat{\theta}^2_v \mu^2(v). \qedhere
\end{align}

\end{document}